\documentclass[12pt,draftcls,onecolumn]{IEEEtran}
\usepackage{cite} 
\usepackage{xspace}
\usepackage{filecontents}

\usepackage{float}
\usepackage{pgf, tikz}

\usepackage{graphicx}

\usepackage[cmex10]{amsmath}

\usepackage{bm}
\usepackage{soul,color} %
\usepackage{amssymb} 

\usepackage{array}

\usepackage{subfigure}
\usepackage{caption2} 

\usepackage{amsfonts}
\usepackage{amsthm}   %

\usepackage{enumerate} %

\usepackage{extarrows}

\usepackage{tkz-fct}

\newtheorem{lemma}{Lemma}
\newtheorem{proposition}{Proposition}
\newtheorem{corollary}{Corollary}

\newtheorem{remark}{Remark}

\usetikzlibrary{patterns,snakes}

\usepackage{enumerate}

\usepackage{extarrows}

\newcommand{\myexpect}[2]{\mathsf{E}_{#1}\left[ #2 \right]}

\newcommand{\myprobability}[1]{\mathrm{Pr}\!\left[ #1 \right]}

\newcommand{\myP}{\mathcal{P}}

\newcommand{\SNR}{\mathsf{SNR}}

\newcommand{\snr}{\mathsf{SNR}}

\newcommand{\Perr}{P_{\text{err}}}
\newcommand{\myd}{\mathrm{d}}
\newcommand{\Hb}{H_{\text{b}}}

\newcommand{\PUB}{P^{\text{UB}}_1}
\newcommand{\PLB}{P^{\text{LB}}_1}

\newcommand{\PUBber}{P^{\text{UB}}_2}
\newcommand{\PLBber}{P^{\text{LB}}_2}

\newcommand{\PUBasyn}{P^{\text{UB,Asyn}}_1}
\newcommand{\PLBasyn}{P^{\text{LB,Asyn}}_1}

\newcommand{\PUBberasyn}{P^{\text{UB,Asyn}}_2}
\newcommand{\PLBberasyn}{P^{\text{LB,Asyn}}_2}

\newcommand{\tr}{\mathsf{tr}}

\begin{document}
\title{\huge Backscatter Multiplicative Multiple-Access Systems: Fundamental Limits and Practical Design}
\author{\IEEEauthorblockN{Wanchun Liu, Ying-Chang Liang, Yonghui Li and Branka Vucetic}
\thanks{\setlength{\baselineskip}{13pt} \noindent The authors are with School of Electrical and Information Engineering, the University of Sydney, Sydney 2006, Australia.
(emails: \{wanchun.liu,\ yc.liang,\ yonghui.li,\ branka.vucetic\}@sydney.edu.au).
Part of the paper has been submitted to the Proc. IEEE ICC 2018.
}
}

\maketitle

\vspace{-1.9cm}
\begin{abstract}
In this paper, we consider a novel \emph{ambient backscatter multiple-access} system,
where a receiver (Rx) simultaneously detects the signals transmitted from an active transmitter (Tx) and a backscatter Tag. 
Specifically, the information-carrying signal sent by the Tx arrives at the Rx through two wireless channels: the direct channel from the Tx to the Rx, and the backscatter channel from the Tx to the Tag and then to the Rx.
The received signal from the backscatter channel also carries the Tag's information because of the \emph{multiplicative} backscatter operation at the Tag.
This multiple-access system introduces a new channel model, referred to as \emph{multiplicative multiple-access channel} (M-MAC).
We analyze the achievable rate region of the M-MAC and prove that its  region is strictly larger than that of the conventional time-division multiple-access scheme in many cases, including, e.g., the high SNR regime and the case when the direct channel is much stronger than the backscatter channel. Hence, the multiplicative multiple-access scheme is an attractive technique to improve the throughput for ambient backscatter communication systems.
Moreover, we analyze
the detection error rates for coherent and noncoherent modulation schemes adopted by the Tx and the Tag, respectively, in both synchronous and asynchronous scenarios, which further bring interesting insights for practical system design.
\end{abstract}
\vspace{-0.6cm}
\section{Introduction}
As described by ``Cooper's Law", the data transactions through wireless
communication over a given area approximately double every two
years since the advent of the cellular phone. 
During the past few years, in addition to human-type traffic, the everlasting demand for wireless communication also comes from \emph{massive} machine-type communication (MTC) in various emerging applications, such as Internet-of-Things, Industry~4.0 and Smart~X~\cite{PetarMag16}.
To enable low power and low-cost MTC, 
the backscatter
communication (BackCom) technique has emerged as a promising solution, which transmits information by simply reflecting the incident radio-frequency (RF) signal~\cite{WanchunMag}.
A BackCom system usually makes use of a dedicated radio source to transmit an \emph{unmodulated/sinusoidal} signal to a backscatter node, which then performs backscatter modulation by varying its reflection coefficient~\cite{Boyer14}.

In the past decade, there have been many studies on designing BackCom systems and networks~\cite{Bletsas09,Almaaitah14,Bletsas14,Katabi12,WanchunTimeHopping}. 
One main research direction is to design multiple-access BackCom networks, in which one sinusoidal-signal-emitting reader serves multiple tags. 
In~\cite{Bletsas09}, the authors proposed a scheme to avoid the transmission collision by jointly adopting directional beamforming at the reader and frequency-shift keying modulation at the tags such that each tag has a higher probability to occupy a unique communication resource. 
In~\cite{Almaaitah14} and~\cite{Bletsas14}, an alternative multiple-access schemes were proposed based on random access and time-division multiple access, respectively. 
In~\cite{Katabi12}, an interesting approach for collision avoidance was proposed based on a compressive-sensing algorithm. 
The authors in~\cite{WanchunTimeHopping} considered a backscatter interference network consisting of multiple reader-tag pairs. To suppress the interference in such a network, a novel time-hopping based scheme was proposed.

Recently, a novel technique, \emph{ambient} BackCom~\cite{Ambient13}, 
has drawn growing interest from both academia and industry due to easy-to-implement communications.
In an ambient BackCom system,
the tag relies on the \emph{modulated} ambient RF signals (e.g., WiFi and cellular signals) rather than unmodulated dedicated ones,
and the receiver receives both the original signal sent from the ambient RF source (e.g., a WiFi access point or a base station) and the backscattered signal, which carries the information of the tag.
In such a system, unlike the conventional BackCom system, to detect the tag's backscattered signal, the receiver also needs to handle the \emph{direct-link interference} from the ambient RF source, which is a new challenge in receiver design introduced by ambient BackCom systems~\cite{GongpuWang2016,FeifeiGao2017,GangYangGC16,GangYangICC17}.

In~\cite{GongpuWang2016} and~\cite{FeifeiGao2017}, 
maximum-likelihood based detection methods were proposed to detect the tag's information, where the direct-link interference is treated as a part of background noise.
However, such methods work well only when the signal strength of the backscattered one is comparable to that of the direct-link interference.
In~\cite{GangYangGC16}, a novel ambient BackCom system was proposed based on ambient orthogonal frequency division multiplexing
(OFDM) carriers. Specifically, the receiver is able to suppress the direct-link interference by exploiting the repeating structure, i.e., the cyclic prefix, of the high rate ambient OFDM signal first, and then detects the low-rate backscatter signal.
In other words, the interference cancellation method is based on the \emph{data-rate asymmetry} between the ambient RF source and the tag.
In~\cite{GangYangICC17}, a multiple-antenna receiver is adopted to cancel the direct-link interference by leveraging the \emph{spatially} orthogonal property of the channels of the backscattered signal and the interference.

In this paper, different from the existing studies focusing on direct-link suppression/interference cancellation, 
we consider an \emph{ambient backscatter multiple-access} system, where the receiver detects both the signals sent from the ambient RF source and the tag. For example, in the smart home application, a smart phone simultaneously receives the Internet data from a WiFi access point and the smart home data from backscatter sensors.
The contributions are summarized as follows:  
\begin{enumerate}
	\item We consider an ambient backscatter multiple-access system which consists of a transmitter (Tx), a backscatter Tag and a receiver (Rx).  
	Instead of treating the direct-link signal as interference as in most existing works, the Rx of the proposed system recovers information from the Tx and the Tag simultaneously. 
	Such a multiple-access system is different from the conventional linear additive multiple-access system because of the multiplicative operation at the tag. 
	Thus, the non-linear property of our system brings new analysis and design challenges.
	Furthermore, different from \cite{GangYangICC17}, in which the receiver relies on the multi-antenna beamforming method to separate the direct-link  signal and the backscattered one, we focus on a more fundamental setting and assume that each node of the system has a single antenna, though the results derived can be readily extended to multi-antenna settings.
	
	\item We further introduce a new multiple-access channel named \emph{multiplicative multiple-access channel} (M-MAC). Specifically, the output of the channel contains two parts: one is the direct-link signal sent from the Tx and the other is the backscatter signal, i.e., the Tag's signal multiplied with the Tx's signal due to the backscatter operation.
	To the best of the authors' knowledge, such a MAC with the multiplicative operation has never been analyzed in the open literature. 
	
	\item We derive the fundamental limit of the M-MAC by analyzing its \emph{achievable rate region}.
	This region is proved to be strictly larger than the rate region achieved by the conventional time-division multiple-access (TDMA) scheme, i.e., the Rx only receives either the information from the Tx or the Tag each time, in many cases.
	Specifically, the M-MAC region is larger than the TDMA region in the high transmission power regime or in the typical case that the direct channel is much stronger than the backscatter channel.
	Our numerical results also show that in a practical range of the SNR and channel conditions, the multiplicative multiple-access scheme always achieves better rate performance than the conventional TDMA scheme.
	Therefore, the proposed multiple-access scheme improves the rate performance of the system in practice.
	
	\item We consider a practical multiplicative multiple-access system, where the Tx and the Tag adopt coherent and noncoherent modulation schemes, respectively, and the Rx does not have full channel state information (CSI) of the Tag. We also consider the scenario that the Tx and the Tag are not perfectly synchronized.
	To detect both the signals of the Tx and the Tag simultaneously, we propose a novel joint coherent and noncoherent detection method.
	Moreover, we propose an analytical framework and comprehensively analyze the detection error rates of signals transmitted from the Tx and the Tag. Our results show some interesting design insights that the asynchronous transmission improves the error performance of the Tx but decreases the performance of the Tag.
\end{enumerate}
Notations: $\myprobability{\mathcal{A}}$ denotes the probability of the event $\mathcal{A}$. 
$\myexpect{}{X}$ denotes the expectation of the random variable $X$.
$\tr\left(M\right)$ and $M^H$ denote the trace and Hermitian transpose of the matrix $M$, respectively.
$\bm I_N$ is the $N \times N$ identity matrix.
$\mathcal{CN}(\cdot,\cdot)$ denotes the complex Gaussian distribution.
$h(\cdot)$, $h(\cdot,\cdot)$, $h(\cdot \vert \cdot)$ denote the differential entropy, joint and conditional differential entropy, respectively. ${I}(\cdot;\cdot)$ denotes the mutual information.
Random variables and their realizations are denoted by upper and lower case letters, respectively.
Vectors and scalers are denoted by bold and normal font letters, respectively.
$o(\cdot)$ is the little o notation.
$Q(\cdot)$ is the Q-function.

\begin{figure*}[t]
	\renewcommand{\captionfont}{\small} \renewcommand{\captionlabelfont}{\small}
	\minipage{0.5\textwidth}
	\centering
\begin{tikzpicture}
\draw  (-3,3) rectangle (-2,2) node (v3) {};
\draw  (-3,0.5) rectangle (-2,-0.5);
\draw  (0.9,1.8) rectangle (2.4,0.8);
\node (v1) at (-2.5,2.5) {Tx};
\node (v2) at (-2.5,0) {Rx};
\node at (1.65,1.3) {Tag ($\rho$)};

\draw [-latex,thick,blue](-2.5,2) -- (-2.5,0.5);
\draw [-latex,thick,blue](-2,2.4) -- (0.9,1.35) node (v4) {};
\draw [-latex,thick,blue](0.9,1.15) -- (-2,0);
\node at (-3.5,2.5) {$X_1$};
\node at (2.9,1.3) {$X_2$};
\node at (-3.5,0) {$Y$};
\node at (-3,1.2) {$g_1$};
\node at (-0.05,1.2) {$g_2$};
\end{tikzpicture}
\vspace{-0.5cm}
\caption{An ambient BackCom system.}
\vspace{-0.3cm}
\label{fig:Sys1}
	\endminipage
	\minipage{0.5\textwidth}
	\centering
	\begin{tikzpicture}
\draw (-3,-2.5) node (v5) {} -- (2,-2.5) node (v6) {};

\draw (-3,-2.5) -- (-3,-2) node (v8) {};
\draw (-2,-2.5) -- (-2,-2);
\draw (-1,-2.5) -- (-1,-2);
\draw (0,-2.5) -- (0,-2);
\draw (1,-2.5) -- (1,-2);
\draw (2,-2.5) -- (2,-2) node (v9) {};

\draw (-3,-2) node (v7) {} -- (2,-2) node (v11) {};

\draw  (-3,-3) rectangle (2,-3.5);
\node (v10) at (-0.5,-1.5) {$\bm X_1$};
\node at (-0.5,-3.3) {$X_2$};

\draw (v7) -- (v10);
\draw (v10) -- (v11);

\node at (-4,-2.3) {Tx:};
\node at (-4,-3.3) {Tag:};
\end{tikzpicture}
\caption{Illustration of $\bm X_1$ and $X_2$ when $N=5$.}
\vspace{-0.5cm}
\label{fig:Sys2}
	\endminipage
	\vspace*{-0.3cm}
\end{figure*}

\vspace{-0.4cm}
\section{System Model}
We consider an ambient backscatter multiple-access system, e.g., for smart-home applications, which consists of an active Tx, a passive Tag and a Rx, as illustrated in Fig.~\ref{fig:Sys1},
each with single antenna.
The Rx receives and detects both the Tx's signal, $X_1$, and the Tag's signal, $X_2$.
$\myP$ is the transmit power of the Tx, and $\rho$ is the reflection coefficient of the Tag\footnote{The $(1-\rho)$ fraction of the incident signal power of the Tag is sent to the RF energy harvester that powers the Tag's circuit.}. 
$g_1$ and $g_2$ are the channel coefficients of the \emph{direct channel}, i.e., the Tx-Rx channel, and the \emph{backscatter channel}, i.e., the Tx-Tag-Rx channel, respectively.
We assume \emph{static channels} (i.e., the channel coherent time is much longer than the symbol durations of the Tx/Tag), and the constant channel state information\footnote{For block-fading channels, the ergodic performance can be derived based on the analysis of this paper.}, $g_1$ and $g_2$, is known at the~Rx.

Since the Tag may have a longer symbol duration, e.g., a lower data rate, than the Tx, we assume that the symbol duration of $X_1$ and $X_2$ are $T$~s and $N T$~s, respectively, where $N \in \mathbb{Z}^+$ is the symbol-length ratio of $X_2$ to $X_1$. 
During the $NT$~s of interest, the transmitted signal of the Tx, $\bm X_1 \triangleq \left[X_{1,1},...,X_{1,N}\right]$, is a vector containing $N$ symbols of the Tx, as illustrated in Fig.~\ref{fig:Sys2}.
To detect both the high data rate signal of the Tx, $X_1$, and the low data rate signal of the Tag, $X_2$, we assume that the Rx samples the received signal every $T$~s.
Also, we assume that both the Tx and the Tag are perfectly synchronized at the~Rx.\footnote{A more practical asynchronous scenario will be discussed in Sec.~\ref{sec:asyn}.}

We consider an average power constraint of the Tx, $\tr\left(\myexpect{}{\bm X^{H} \bm X }\right) \leq N$.
In addition, due to the cost and implementation complexity constraint of the Tag, we assume that the constellation of the Tag is $\left\lbrace c_1, c_0 \right\rbrace$, which only has two complex number elements. 
Moreover, due to the passive backscatter property, we further assume that $\vert c_{1} \vert$ and $\vert c_{0} \vert \in \left[0,1\right]$.
For instance, most of the commercially available off-the-shelf tags use binary modulation, e.g., the binary-phase-shift keying (BPSK) or on/off modulation schemes with the constellation $\left\lbrace 1, -1 \right\rbrace$ or $\left\lbrace 1, 0 \right\rbrace$, respectively.

The received signal of the Rx during one symbol duration of the Tag can be written as
\begin{equation}
\bm Y'= g_1 \sqrt{\myP} \bm X_1 + g_2 \sqrt{\rho} \sqrt{\myP} X_2 \bm X_1 +\bm  Z',
\end{equation}
where $\bm  Z'$ is the received AWGN, and $\bm Z'\sim \mathcal{CN}(0,\sigma'^2 \bm I_N)$.
After normalization and simplification, the received signal, $\bm Y$, can be represented as
\begin{align}
\label{M-MAC channel}
\bm Y
&= \sqrt{\myP} \bm X_1 + g \sqrt{\rho} \sqrt{\myP} \bm X_1 X_2 +\bm  Z,\\
\label{M-MAC channel2}
&= \left(1 + g \sqrt{\rho} X_2 \right) \sqrt{\myP} \bm X_1 +\bm  Z,
\end{align}
where $g \triangleq g_2/g_1 = \vert g \vert e^{j \theta}$, $\bm Z\sim \mathcal{CN}(0,\sigma^2 \bm I_N)$ and $\sigma \triangleq {\sigma'}/{\vert g_1 \vert }$.
Specifically, $g$ and $\theta$ are the coefficient and phase of the \emph{relative backscatter channel}.
For brevity, we define the SNR of the multiple-access system, $\snr$, as 
\vspace{-0.3cm}
\begin{equation} \label{snr}
\snr \triangleq \myP/\sigma^2.
\end{equation}

From an information-theoretic perspective, \eqref{M-MAC channel} can be named as the multiplicative multiple-access channel (M-MAC).

\begin{remark}
	It is clear that the received signal of the M-MAC contains two information related terms: one is directly related to the Tx's signal, and the other is a  multiplicative form of the signals of the Tx and the Tag.
	Conventional additive MAC has been extensively investigated in the open literature, while the M-MAC that originates from the ambient BackCom, has not been systematically analyzed yet.
\end{remark}

\vspace{-0.4cm}
\section{Achievable Rate Region of M-MAC} \label{sec:M-MAC region}
In the section, we investigate the \emph{achievable rate region} of the M-MAC.
Given the joint distribution of $\bm X_1$ and $X_2$, say $\bm p(\bm X_1,X_2)$,  we have the following achievable rate region, $\mathcal{R}(\bm X_1, X_2)$, which is the set of $(R_1, R_2)$ such  that~\cite{BookInfo}
\begin{align} \label{rate-region}
R_1 \leq I(\bm X_1;\bm Y \vert X_2),\ 
R_2 \leq I(X_2;\bm Y \vert \bm X_1),\ 
R_1 + R_2 \leq I(\bm X_1,X_2;\bm Y),
\end{align}
where $R_1$ and $R_2$ are the achievable rates of the Tx's information and the Tag's information, respectively, during $NT$~s.


For mathematical simplicity, we assume that $\bm X_1$ follows the Gaussian distribution, i.e., $\bm X_1\sim \mathcal{CN}(0, \bm I_N)$, and $X_2$ follows the uniform distribution, i.e., $\myprobability{X_2 = c_1} = \myprobability{X_2 = c_0} = 1/2$.
Also, $\bm X_1$ and $X_2$ are independent.
With these assumptions, we can calculate the maximum achievable rates of $\bm X_1$ and $X_2$, and the maximum achievable sum rate, i.e., $I(\bm X_1;\bm Y \vert X_2)$, $I(X_2;\bm Y \vert \bm X_1)$ and $I(\bm X_1,X_2;\bm Y)$, in sequence.

\subsection{Maximum Achievable Rate of $\bm X_1$}
Let us assume that $X_2$ has been successfully detected by the Rx. We obtain the maximum achievable rate of $\bm X_1$ as follows:
\begin{align}
I(\bm X_1;\bm Y \vert X_2) 
&= h(\bm Y\vert X_2)-h(\bm Y\vert \bm X_1, X_2)\\
&=h\left(\left(1 + g \sqrt{\rho} X_2 \right) \sqrt{\myP} \bm X_1 + \bm Z\vert X_2\right)-h(\bm Z)\\
&=\frac{1}{2}\!\left(\!h\left(\!\left(\!1 + g \sqrt{\rho} c_1 \!\right) \sqrt{\myP}  \bm X_1 + \bm Z \!\right) 
+ h\left(\!\left(\!1 + g \sqrt{\rho} c_{0} \right) \sqrt{\myP}  \bm X_1 + \bm Z \!\right)\!\right) 
\label{firstI}
-h(\bm Z). 
\end{align}
Since both $\bm X_1$ and $\bm Z$ are vectors of independent and identically distributed (i.i.d.) complex Gaussian random variables, we have~\cite{BookInfo}
\begin{equation}
\label{h}
h\!\left(\!\left(1 + g \sqrt{\rho} c_i \right) \sqrt{\myP} \bm X_1 + \bm Z \!\right) 
\!=\! N \log_2\left(\!\pi e \left(\!\vert 1 + g \sqrt{\rho} c_i \vert^2\myP+\sigma^2\right)\!\right),\ \!
h\!(\bm Z)\!=\! N \log_2\left(\!\pi e \sigma^2\right),\ i=0,1.
\end{equation}
Taking \eqref{h} into \eqref{firstI}, we have
\vspace{-0.3cm}
\begin{equation} \label{maxX1}
I(\bm X_1;Y \vert X_2) = \frac{1}{2} \left(h_1 +h_0\right),
\end{equation}
where
\vspace{-0.3cm}
\begin{equation}\label{hnew}
h_i \triangleq N \log_2\left(1 + \frac{\vert 1 + g \sqrt{\rho} c_i \vert^2\myP}{\sigma^2}\right),\ i=0,1.
\end{equation}

\subsection{Maximum Achievable Rate of $X_2$} \label{sec:R2}
Let us assume that $\bm X_1$ has been successfully detected by the Rx. We can obtain the maximum achievable rate of $X_2$.

From \eqref{M-MAC channel}, it is easy to see that the optimal way to detect $X_2$ is to, firstly, remove the known additive interference, $\sqrt{\myP} \bm X_1$, and then perform maximal ratio combining (MRC) using $\bm X_1$. Thus, the received signal after MRC and normalization, $\tilde{Y}$, can be written as
\begin{equation}\label{binary AWGN channel}
\tilde{Y} = X_2 +\tilde{Z}, 
\end{equation}
where $\tilde{Z} \sim \mathcal{CN}(0,\tilde{\sigma}^2)$, and 
$\tilde{\sigma}^2 \triangleq \frac{\sigma^2}{\vert g \vert^2 \rho \myP \vert \bm X_1 \vert^2 }$.
Since $X_2$ is a binary input and $\tilde{Z}$ is the continuous additive Gaussian noise, \eqref{binary AWGN channel} is the \emph{binary input AWGN channel} with the input $X_2$ and the output $\tilde{Y}$~\cite{BAWGN}. 
Therefore, the maximum achievable rate of $X_2$ can be written as
\begin{align}\label{I2}
I(X_2;\bm Y \vert \bm X_1) &= \myexpect{\bm X_1}{I(X_2;\bm Y )\vert \bm X_1}= \myexpect{\bm X_1}{I(X_2; \tilde{Y} )\vert \bm X_1},
\end{align}
\vspace{-0.3cm}
where~\cite{BAWGN}
\begin{equation} \label{IX2tildeY}
\begin{aligned}
I(X_2; \tilde{Y})
&= h(\tilde{Y})-h(\tilde{Z})
= - \int_{-\infty}^{\infty} \Psi\left(\tilde{y},\tilde{\sigma}^2\right) \log_2\left(\Psi\left(\tilde{y},\tilde{\sigma}^2\right)\right) \myd \tilde{y} - \frac{1}{2}\log_2\left(\pi e \tilde{\sigma}^2\right),
\end{aligned}
\end{equation}
\begin{equation}
\begin{aligned}
\Psi\left(\tilde{y},\tilde{\sigma}^2\right)
&= \frac{1}{2\sqrt{\pi \tilde{\sigma}^2}} \left(\exp\left(- \frac{\left(\tilde{y}-\vert c_1-c_0 \vert/2\right)^2}{\tilde{\sigma}^2}\right)
+
\exp\left(- \frac{\left(\tilde{y}+\vert c_1-c_0 \vert/2\right)^2}{\tilde{\sigma}^2}\right)
\right),
\end{aligned}
\end{equation}
and $\tilde{\sigma}^2$ is defined under~\eqref{binary AWGN channel}.

Although $I(X_2;\bm Y \vert \bm X_1)$ in~\eqref{I2} does not have a closed-form expression, a closed-form lower bound is derived as follows:
\begin{align}
I(X_2;\bm Y \vert \bm X_1) &= \myexpect{\bm X_1}{I(X_2; \tilde{Y} )\vert \bm X_1} \notag \\
\label{binary decision}
& \geq \myexpect{\bm X_1}{I(X_2; \mathbb{B}(\tilde{Y}) )\vert \bm X_1}\\
\label{binary channel}
&= \myexpect{\bm X_1}{1-\Hb\left(\Perr\right)}
\end{align}
\begin{align}
\label{jensen1}
&> 1- \Hb\left(\myexpect{\bm X_1}{\Perr}\right)\\
\label{tse_equation}
&=\! 1 \!-\! \Hb\!\left(\!\!
\left(\frac{1-\mu}{2}\right)^{\!\!\!\!N} \sum_{i=0}^{N\!-\!1}\!\! {N-1+i \choose i}\!\!\!\left(\!\frac{1+\mu}{2}\!\right)^{\!\!i}
\right)\\
&\triangleq
\underline{H},
\end{align}
where \eqref{binary decision} is due to the data-processing inequality~\cite{BookInfo}, and $\mathbb{B}(\cdot)$ is the binary decision function that directly maps $\tilde{Y}$ to $c_1$ or $c_0$ when $\vert \tilde{Y}-c_1 \vert \leq \vert \tilde{Y}-c_0 \vert$ or $\vert \tilde{Y}-c_1 \vert > \vert \tilde{Y}-c_0 \vert$, respectively.
Thus, $X_2 \rightarrow \mathbb{B}(\tilde{Y})$ is a binary symmetric channel that has a binary input and a binary output.
The mutual information of such a binary symmetric channel is
$1-\Hb\left(\Perr\right)$ given in \eqref{binary channel}, where
$\Hb\left(\cdot \right)$ is the binary entropy function, $\Hb\left(p\right) \triangleq - p \log_2(p) - (1-p) \log(1-p)$ and $p \in \left[0,1\right]$,
and $\Perr$ is the error probability of the binary detection, which is given by
\begin{equation}
\Perr 
= Q\left(\sqrt{\frac{\vert g \bm X_1 \vert^2  \rho \myP}{2 \sigma^2} }  \vert c_1-c_{0}\vert\right).
\end{equation}
\eqref{jensen1} is based on the Jensen's inequality as the function $1- \Hb(\cdot)$ is strictly convex. 
\eqref{tse_equation} is directly obtained from~\cite[(3.37)]{BOOKTse}, and 
\begin{equation} \label{defi:mu}
\mu = \sqrt{\frac{\vert g\vert^2  \rho \myP \vert \frac{c_1-c_{0}}{2}\vert^2}{\sigma^2 +  \vert g\vert^2 \rho \myP \vert \frac{c_1-c_{0}}{2}\vert^2}}.
\end{equation}

\subsection{Maximum Achievable Sum Rate}
The maximum achievable sum rate is given~by
\begin{align}
\label{X1X1Y}
I(\bm X_1,X_2;\bm Y) 
\!=\! h(\bm Y) \!-\! h(\bm Y \vert \bm X_1, X_2)
\!=\!h(\bm Y)\!-\!h(\bm Z).
\end{align}
Based on \eqref{M-MAC channel2} and the binary distribution of $X_2$, it is easy to obtain that the probability density function (PDF) of the complex random variable $\bm Y$, which is given by
\begin{equation}
f_{\bm Y}(\bm y) = \frac{1}{2} \left(f_1(\bm y) + f_{0}(\bm y)\right),
\end{equation}
where  $f_i(\bm y)$ is the PDFs of 
the random variable, $\left(1 + g \sqrt{\rho} c_i \right) \sqrt{\myP} \bm X_1 +\bm  Z$ following the distribution 
$\mathcal{CN}\left(0, \left(\vert 1 + g \sqrt{\rho} c_i \vert^2 \myP + \sigma^2\right) \bm I_N\right)$, $i=0,1$.
Thus, we have 
\begin{align}
\label{hY}
h(\bm Y)& = - \int_{\bm Y} f_{\bm Y}(\bm y) \log_2\left(f_{\bm Y}(\bm y)\right) \myd \bm y\\
\label{jensen}
& \geq  \!-\frac{1}{2}\int_{\bm Y} \! f_1(\bm y) \log_2\left(f_1(\bm y)\right) \myd \bm y
\!-\frac{1}{2}\int_{\bm Y} \! f_{0}(\bm y) \log_2\left(f_{0}(\bm y)\right) \myd \bm y,
\end{align}
where \eqref{jensen} is due to the convex function, $x \log_2(x)$, and the Jensen's inequality, and the equality holds when $f_1(\bm y) = f_{0}(\bm y)$ for all $\bm y\in \mathbb{C}^{N}$, i.e., $\vert 1 + g \sqrt{\rho} c_1 \vert
=
\vert 1 + g \sqrt{\rho} c_0 \vert$.

Taking \eqref{hY} and \eqref{h} into \eqref{X1X1Y}, the maximum achievable sum rate is obtained. Further, taking \eqref{jensen} and \eqref{h} into \eqref{X1X1Y}, its closed-form lower bound is obtained as follows:
\begin{align} \label{sum rate lower bound}
I(\bm X_1,X_2;\bm Y) \geq \underline{h} \triangleq \frac{1}{2} \left(h_1 + h_0 \right),
\end{align}
where $h_1$ and $h_0$ are defined in \eqref{hnew}.

\section{Time-Division Multiple Access vs. Multiplicative Multiple Access}
In the section, 
we consider a comparative study of the conventional time-division multiple-access channel/system and the multiplicative multiple-access channel/system, and study the strict convexity of the achievable rate region of the M-MAC.

\subsection{The Rate Region of the Conventional TDMA Backscatter Channel} \label{sec:time sharing}
In the conventional TDMA backscatter system, there are two phases in information transmissions: 1) the Tx transmits its information to the Rx while the Tag is able to adjust to the best reflection coefficient
so that the Tag's reflected signal is constructively combined with the direct-link signal at the Rx; 2) the Tx simply emits a deterministic sinusoidal signal with power $\myP$ while the Tag sends its information using backscatter modulation. Note that in each phase only one node is able to transmit information.
The received signals of the Tx and the Tag in phases 1 and~2 are given in \eqref{M-MAC channel2} and \eqref{binary AWGN channel}, respectively.
Also, it can be proved that the maximum rates of the Tx and the Tag are achieved in phase~1 and~2, respectively, 
because of zero inter-user-interference.

In phase 1, the maximum rate of $\bm X_1$ is described by the following proposition.
\begin{proposition} \label{prop:time-sharing 1}
	\normalfont
	The maximum rate of $\bm X_1$, $R^{\max}_1$, is achieved when the Tag transmits a constant symbol which makes the received signal power largest at the Rx, and $R^{\max}_1$ is given by  
	\begin{equation} \label{h large}
	R^{\max}_1
	=\overline{h}
	\triangleq \max \left\lbrace
	h_1,	h_0
	\right\rbrace,
	\end{equation}
	where $h_1$ and $h_0$ are defined in \eqref{hnew}.
\end{proposition}

In phase 2, the maximum rate of $X_2$ and its upper bound are described by Proposition~\ref{prop:time-sharing2}.
\begin{proposition}\label{prop:time-sharing2}
	\normalfont	
	The maximum rate of $X_2$, $R^{\max}_2$, is achieved when the Tx transmits the deterministic sinusoidal signal with power $\myP$. $R^{\max}_2$ is given in \eqref{IX2tildeY} with the parameter
	\begin{equation}
	\tilde{\sigma}^2 = \frac{\sigma^2}{N \vert g \vert^2 \rho \myP}.
	\end{equation}
	The closed-form upper bound of $R^{\max}_2 $ is given by
	\begin{align} \label{time sharing upperbound}
	{R}^{\max}_2 < \overline{H} \triangleq \min\left\lbrace \log_2\left(1 + \frac{N \vert g \vert^2\rho \myP}{\sigma^2}\right), 1 \right\rbrace.
	\end{align}	
\end{proposition}
\begin{proof}
See Appendix~A.
\end{proof}

Therefore, the achievable rate region of the conventional TDMA channel is the \emph{triangle} with vertexes $(0,0)$, $(0,{R}^{\max}_1)$ and $({R}^{\max}_2,0)$.
This region is contained by an achievable region of the M-MAC, since 
both phase 1 and phase 2 of the TDMA scheme are special cases of the multiplicative multiple-access scheme.
However, it is not clear whether the M-MAC has a strictly larger achievable rate region or not compared with the \emph{triangular} TDMA one. 
Equivalently, it is interesting to see whether the achievable rate region of the M-MAC is strictly convex or not.

\subsection{Strict Convexity of the Achievable Rate Region of the M-MAC} \label{sec:special case regions}
Since the achievable rate regions of the M-MAC (in Sec.~\ref{sec:M-MAC region}) and the TDMA channel (in Sec.~\ref{sec:time sharing}) do not have closed-form expressions, a direct proof of the strict convexity of the achievable rate region of the M-MAC is not possible. In the following, we tackle the problem resorting to the inequalities \eqref{jensen1}, \eqref{sum rate lower bound}  and \eqref{time sharing upperbound}.

Before proceeding further, we need the following definitions: $o = (0,0)$, $B_1 = (0,\overline{h})$, $A_1 = (0, \underline{h})$, $B_2=(\overline{H},0)$, $A_2 = (\underline{H},0)$, $D = (R^{\max}_2,0)$, $C=(\underline{H}, \underline{h}-\underline{H})$, as illustrated in Fig.~\ref{Region}.
\begin{figure}
	\renewcommand{\captionfont}{\small} \renewcommand{\captionlabelfont}{\small}	
	\centering
	\begin{tikzpicture}[scale=1.8]
	\draw [->,thick](0,0) -- (3,0);
	\draw [->,thick](0,0) -- (0,3.35);
	\node [fill,circle,scale=0.4,blue] (v1) at (0,2.5) {};
	\node [fill,circle,scale=0.4,blue] (v2) at (1.55,0) {};
	\node [align=center] at (-0.45,2.3) {$A_1$\\$(0,\underline{h})$};
	\node (v4) at (0.85,-0.35) {$A_2$};
	\node at (0,3.6) {$R_1$};
	\node at (3.3,0) {$R_2$};
	\draw [ultra thick,red](v1) -- (1.55,1);
	
	\draw [ultra thick,red](v2) -- (1.55,1);
	
	\node [fill,circle,scale=0.4,blue] (v3) at (0,2.8) {};
	\node [align=center] at (-0.45,3) {$(0,\overline{h})$\\ $B_1$};
	\node (v6) at (2.75,-0.35) {$B_2$};
	\node [fill,circle,scale=0.4,blue] at (1.55,1) {};
	
	\draw [ultra thick,black](v3) -- (1.55,1) -- (2,0);
	\node [fill,circle,scale=0.4,blue] (v5) at (2,0) {};
	
	\node at (1.35,0.95) {$C$};
	\node at (-0.15,-0.15) {$o$};
	\draw [densely dashed,thick] (0,2.8) node (v7) {} .. controls (1.5,2.5) and (1.75,1.4) .. (1.8,0);
	\draw  plot[smooth, tension=.7] coordinates {(0.85,2.45) (1.05,2.8) (1.3,3)};
	\node [right] at (1.3,3.1) {M-MAC capacity region};
	\node [fill,circle,scale=0.4,blue] (v8) at (1.8,0) {};
	\node at (1.75,-0.35) {$D$};
	\node at (0.85,-0.7) {$(\underline{H},0)$};
	\node at (1.75,-0.7) {$(R^{\text{max}}_2,0)$};
	\node at (2.8,-0.7) {$(\overline{H},0)$};
	\draw (1,-0.5);
	\draw (v4) -- (1.45,-0.05);
	\draw (2.1,-0.05) -- (v6);
	\draw [densely dotted,thick] (v7) -- (v8);
	\draw (0.9,1.4) .. controls (0.65,1.25) and (0.55,1.1) .. (0.5,0.75);
	\node [align=center] at (0.7,0.45) {TDMA\\rate region};
	\end{tikzpicture}
	\caption{An illustration of the achievable rate region of the M-MAC. The red-solid-line region, $o-A_1-C-D$, is contained by the achievable rate region of the M-MAC. The black-solid-line region, $o-B_1-C-B_2$, is the region of interest for the analysis.}
	\vspace{-0.5cm}
	\label{Region}
\end{figure}
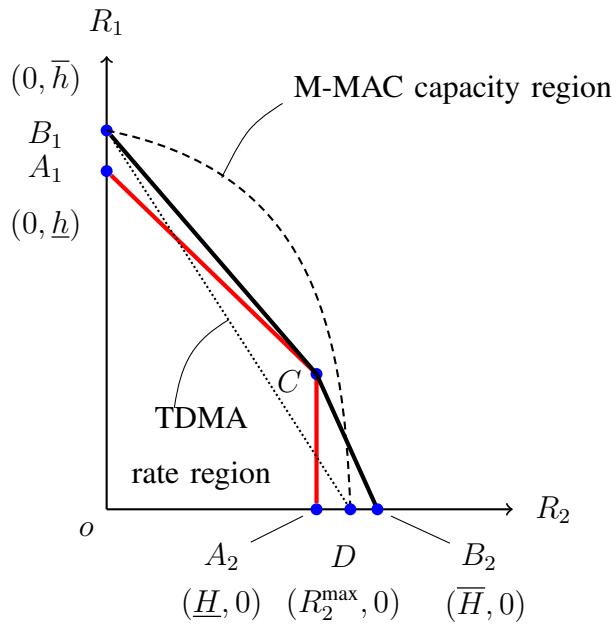

From Propositions~\ref{prop:time-sharing 1} and~\ref{prop:time-sharing2}, \emph{the triangle $o-B_1-D$ is the achievable rate region of the conventional TDMA channel and is contained by the region $o-B_1-B_2$.}
From \eqref{maxX1}, \eqref{jensen1} and \eqref{sum rate lower bound} in Sec.~\ref{sec:M-MAC region}, it can be proved that \emph{the region $o-A_1-C-A_2$ is contained by the achievable rate region of the M-MAC with the Gaussian input, $\bm X_1$, and the binary input,~$X_2$.}

If the achievable rate pair of the M-MAC, the node~$C$, lies outside the region $o-B_1-B_2$, $C$ lies outside the achievable rate region of the TDMA channel, i.e., $o-B_1-D$. 
Hence, \emph{the achievable rate region of the M-MAC}, $o-B_1-C-D$, is strictly convex.
Thus, we have the following lemma.
\begin{lemma}
	\normalfont
	A sufficient condition that the achievable rate region of the M-MAC is strictly convex, is that the region $o - B_1 - C- B_2$ is strictly convex.
\end{lemma}

Furthermore, we obtain the following lemma straightforwardly based on the definitions of $B_1$, $C$ and $B_2$.
\begin{lemma}
	\normalfont
	The rate region closed by $o-B_1-C-B_2$ is strictly convex iff $r_1 > r_2$, 
	where $r_1$ and $r_2$ are the absolute values of the slopes of $B_1-B_2$ and $B_1-C$, respectively, which are defined as
\vspace{-0.4cm}
	\begin{equation}
	\label{r1}
	\begin{aligned} 
	r_1 \triangleq \frac{\ \ \overline{h}\ \ }{\overline{H}},\ 
	r_2 \triangleq 1 + \frac{\overline{h} - \underline{h}}{\underline{H}},
	\end{aligned}	
	\end{equation}
	and $\overline{h}$, $\underline{h}$, $\overline{H}$ and $\underline{H}$ are given in \eqref{h large}, \eqref{sum rate lower bound}, \eqref{time sharing upperbound} and \eqref{tse_equation}, respectively.
\end{lemma}

Based on the expressions of $\overline{h}$, $\underline{h}$, $\overline{H}$ and $\underline{H}$, it is still not clear whether $r_1$ is greater than $r_2$ or not. Therefore, we further investigate the following special cases.

\subsubsection{The High SNR Scenario}
In the high SNR scenario, i.e., letting $\snr \rightarrow \infty$, we have the following result.
\begin{proposition} \label{prop:high snr}
	\normalfont
	In the high SNR scenario, the achievable rate region of the M-MAC is strictly convex.
\end{proposition}

\begin{proof}
See Appendix B.
\end{proof}

\subsubsection{The Weak Backscatter Channel Scenario}
We further investigate the typical case that the direct channel is much stronger than the backscatter channel.
\begin{proposition} \label{prop:typical}
	\normalfont
	In the typical case that the channel power gain of the direct channel is much stronger than that of the backscatter channel, i.e., $\vert g \vert^2 \ll1$, the achievable rate region of the M-MAC is strictly convex.
\end{proposition}
\begin{proof}
	See Appendix~C.
\end{proof}

\subsubsection{The BPSK Scenario}
We also investigate the scenario that the Tag has the BPSK constellation, and have the following results.
\begin{proposition} \label{prop:BPSK-low-snr}
	\normalfont
	Assuming BPSK modulation scheme, in the low SNR scenario, i.e., $\snr \rightarrow 0$, the achievable rate region of the M-MAC is strictly convex.
\end{proposition}
\begin{proof}
	See Appendix~D.
\end{proof}

\begin{proposition}\label{prop:bpsk}
	\normalfont
	In the BPSK scenario, assuming that the phase of the relative backscatter channel, $\theta$, is equal to ${\pi}/{2}$, the achievable rate region of the M-MAC is strictly convex.
\end{proposition}
\begin{proof}
See Appendix E.
\end{proof}

\subsection{Numerical Results}
In this section, we present numerical results for an achievable rate region of the M-MAC, i.e., the polygon $o-B_1-C-D$. Recall that the nodes $B_1$, $D$ and $C$ are achievable rate pairs of the M-MAC as discussed in Sec.~\ref{sec:time sharing} and Sec.~\ref{sec:special case regions}. The region is plotted using the definitions in Sec.~\ref{sec:special case regions}.
Also, we consider the scenario that the Tag adopts the BPSK modulation scheme and $N=1$.
Unless otherwise stated, we set $\snr = 10$~dB, $\rho=0.5$, $\vert g \vert^2 = 0.1$ and $\theta = \pi/4$.
Recall that the achievable rate region of the conventional TDMA channel is the triangle, $o-B_1-D$.
Thus, the proposed multiple-access scheme is able to significantly improve the system performance if we see that the polygon $o-B_1-C-D$ is more like a rectangular instead of a triangle.

In Fig.~\ref{fig:region with diff g}, the achievable rate region is plotted for different power gain of the relative backscatter channel, $\vert g \vert^2$.
We see that the rate region enlarges monotonically with $\vert g \vert^2$.
Therefore, a strong backscatter channel is beneficial to the rate performance of both the Tx and the Tag.
Moreover, for a practical range of $\vert g \vert^2$, i.e., $\vert g \vert^2 \in (0,1]$, 
the region is always strictly convex, which means
the multiplicative multiple-access scheme is always able to improve the rate performance.
Also, we see that the shape of the rate region is more like a rectangle when $\vert g \vert^2$ is very small, e.g., $\vert g \vert^2 = 0.01$, while the shape of the rate region is more like a triangle when $\vert g \vert^2$ is large, e.g., $\vert g \vert^2 = 1$.
This shows that the multiplicative multiple-access scheme improves the rate performance of the system significantly compared with the conventional TDMA scheme in the typical scenario that the direct channel is much stronger than the backscatter channel.

\begin{figure*}[t]
	\renewcommand{\captionfont}{\small} \renewcommand{\captionlabelfont}{\small}
	\minipage{0.5\textwidth}
	\centering
	\includegraphics[scale=0.55]{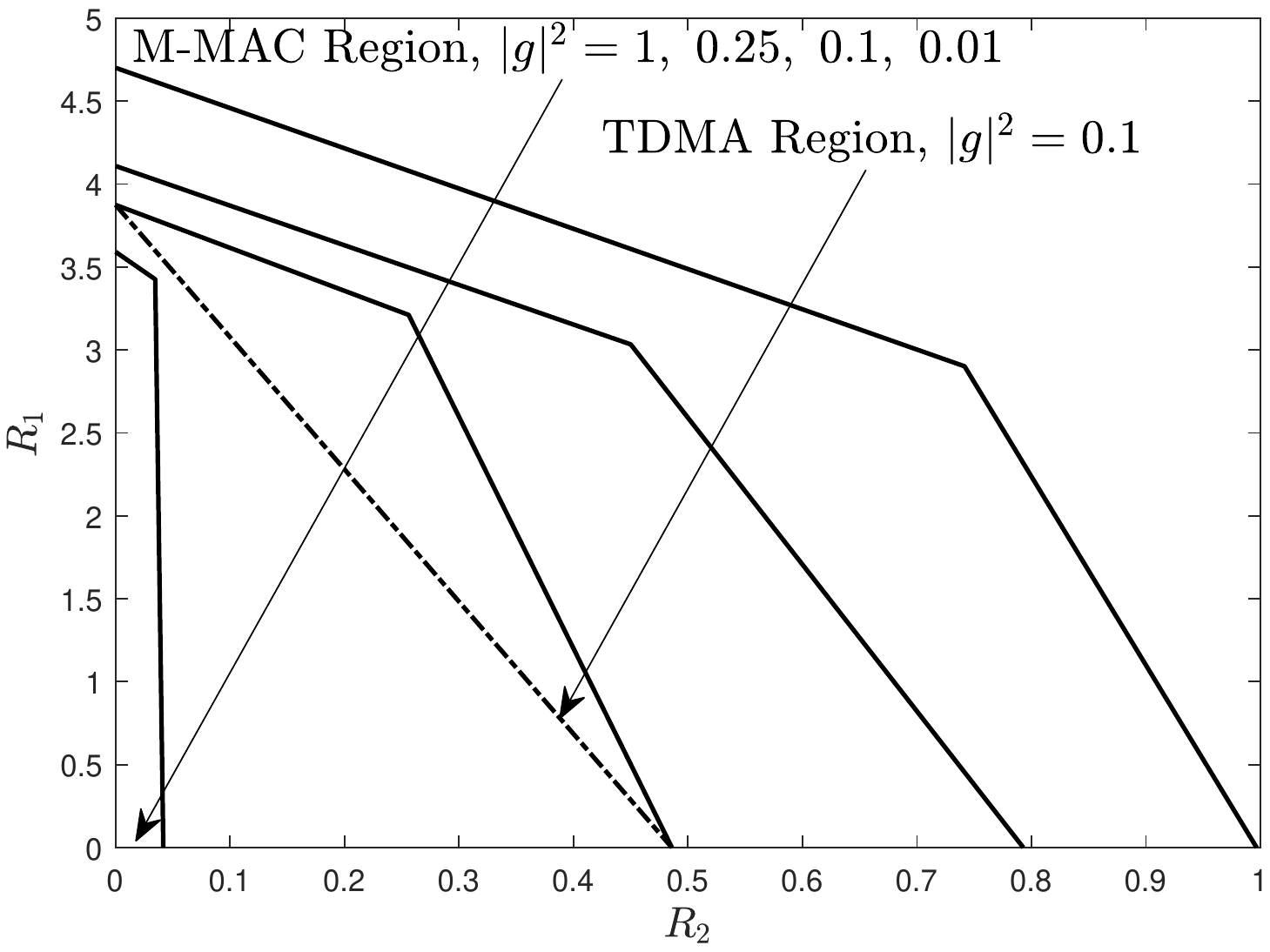}
\vspace{-0.9cm}	
\caption{Achievable rate region with different $\vert g \vert^2$.}
\vspace{-0.3cm}
\label{fig:region with diff g}
	\endminipage
	\minipage{0.5\textwidth}
	\centering
	\includegraphics[scale=0.55]{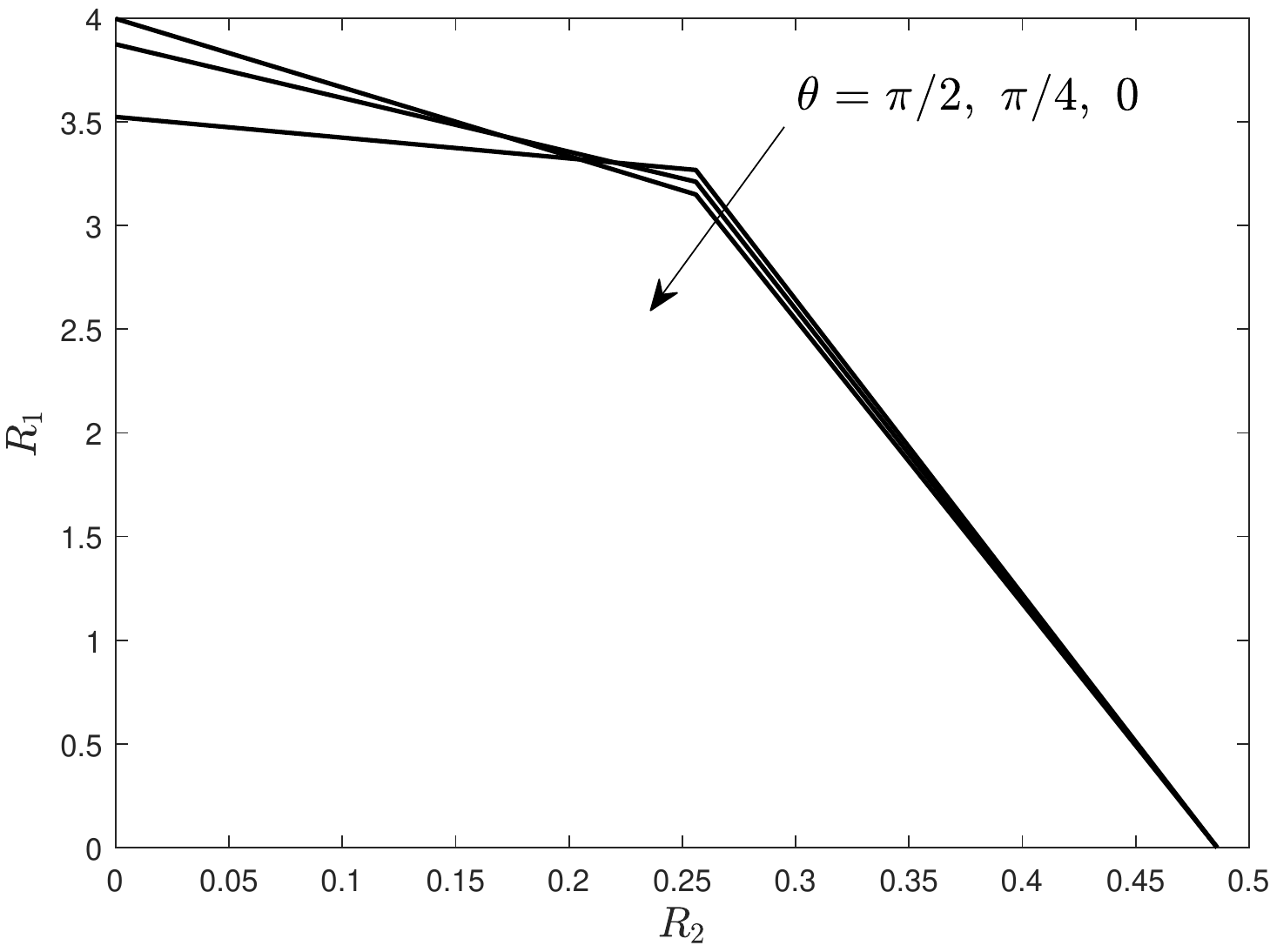}
\vspace{-0.9cm}		
\caption{Achievable rate region with different $\theta$.}
\vspace{-0.3cm}
\label{fig:region with diff theta}
	\endminipage
	\vspace*{-0.0cm}
\end{figure*}

In Fig.~\ref{fig:region with diff theta}, the achievable rate region is plotted for different phase of the relative backscatter channel, $\theta$.
We see that the shape of the rate region changes with $\theta$. 
Specifically,
the achievable rate of $\bm X_1$, $R_1$, is larger when we have a smaller $\theta$ and the rate requirement of the Tag is small (e.g., $R_2 < 0.1$);
the achievable rate of $\bm X_1$ is larger when we have a larger $\theta$ (e.g., $\theta = \pi/2$) and the rate requirement of the Tag is large (e.g., $R_2 > 0.25$).

In Fig.~\ref{fig:region with diff snr},  the achievable rate region is plotted for different $\snr$. 
We see that the achievable rate region enlarges with $\snr$.
Moreover, the region is strictly convex in the practical range of the SNR, i.e., $\snr\in [0,30]$~dB, 
which means that the multiplicative multiple-access scheme is always able to improve the rate performance of the system. 
Also, it is observed that the shape of the rate region is more like a rectangle when $\snr$ is very large, e.g., $\snr = 30$~dB.
In other words, the multiplicative multiple-access scheme has a significant rate improvement of the system compared with the conventional TDMA scheme in the high SNR scenario.

\begin{figure}[t]
	\renewcommand{\captionfont}{\small} \renewcommand{\captionlabelfont}{\small}	
	\centering
	\includegraphics[scale=0.6]{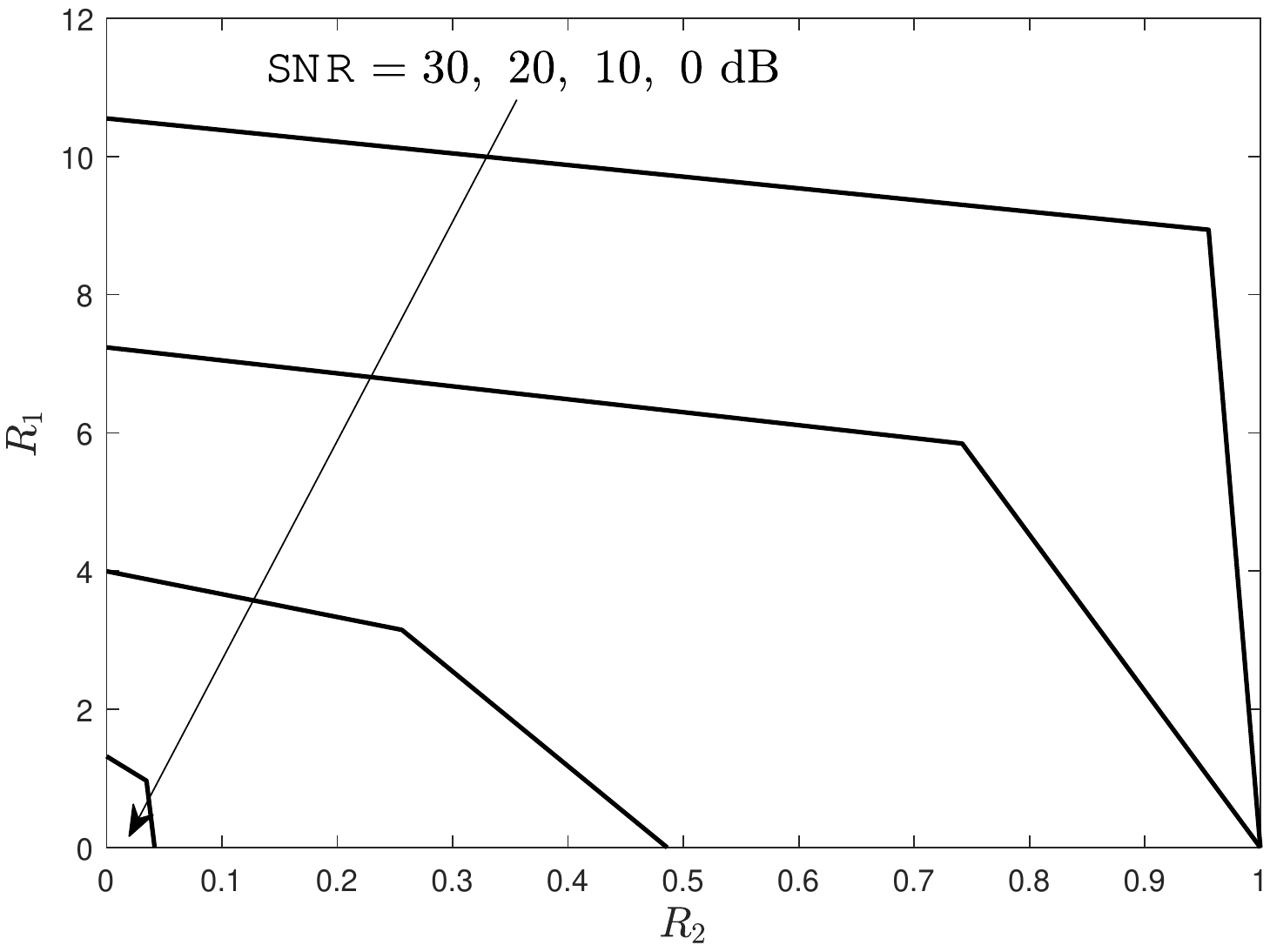}
	\caption{Achievable rate region with different $\snr$.}
	\vspace{-0.1cm}
	\label{fig:region with diff snr}
\end{figure}

\section{Detection Design and Analysis of Multiplicative Multiple-Access Systems}\label{sec:asyn}
In the section, we focus on a practical multiple-access communication based on the following assumptions:

\begin{itemize}
	\item The Tx adopts a commonly used practical $M$-ary modulation scheme with finite alphabet rather than an ideal Gaussian signal. The symbols of the Tx are i.i.d.. The CSI of the Tx is fixed and perfectly known at the Rx.
	\item The Tag adopts a noncoherent modulation scheme, i.e., only the amplitude of its backscattered signal carries information. Since the magnitude of the backscatter channel coefficient often varies much slower than its phase, we assume that only partial CSI of the backscatter channel, i.e., only the magnitude information, is known and fixed at the Rx. We assume that the phase of the channel coefficient keeps constant within each Tag symbol duration and varies symbol-by-symbol. The phase, $\Theta$, follows uniform distribution with the support $[0,2\pi)$.
	\item The Tx is perfectly synchronized with the Rx, however, the Tag may or may not be perfectly synchronized with the Rx depending on its clock accuracy, microcontroller complexity and circuit operating frequency. Therefore, the signals sent from the Tx and the Tag may or may not arrive at the Rx in a symbol-level alignment manner referring to the \emph{synchronous and asynchronous scenarios}, respectively.
	\item The direct channel power gain is much stronger than that of the backscatter channel, i.e., $\vert g \vert^2  \ll 1$.
\end{itemize}

\subsection{Joint Detection Method} \label{sec:detection method}
For the Rx design of the multiplicative multiple-access system, the main challenge is how to effectively detect both the signals sent from the Tx and the Tag.
It is well-known that the maximum-likelihood (ML) method is the optimal, i.e.,
\begin{equation}
\left\lbrace  \hat{\bm x}_1, \hat{x}_2 \right \rbrace = \arg\max_{\bm x_1, x_2} f_{\bm Y}(\bm y \vert \bm x_1, x_2),
\end{equation}
where $\hat{\bm x}_1$ and $\hat{x}_2$ are detected symbols of the Tx and the Tag, respectively, and $f_{\bm Y}(\bm y \vert \bm x_1, x_2)$ is the conditional joint probability density function of $\bm Y$, which takes into account the uncertainty of the noise and the backscatter channel.
However, the decoding complexity of such a method is very high if not impossible, especially when the Tx adopts a higher-order modulation scheme and the length of the sequence $\bm Y$, i.e., $N$, is very large.
For example, assuming $M$-ary modulated $X_1$, the detection complexity of the sequence $\bm Y$ has an order of $\mathcal{O}\left(M^N\right)$.
Thus, we aim to design a low-complexity decoding method that is able to provide a desirable performance.

The decoding method contains two steps:
\begin{itemize}
	\item Detection of $X_1$. The main reason for detecting $X_1$ first is that $X_1$ is a much stronger signal than $X_2$ in practice.	
	Due to the unknown phase of the channel coefficient, $\theta$, and its uniform distribution, the Rx can simply make the detection of $X_1$ using the conventional decision region of the $M$-ary modulation scheme in a symbol-by-symbol manner.

    For example, if $X_2 = 0$, $X_2$ has zero interference in the detection of $X_1$.
    If $X_2 \neq 0$, the interference may occur.	
	Specifically, assuming QPSK and on/off modulation schemes for the Tx and the Tag, respectively, different phase of the relative backscatter channel coefficient results in different received constellation at the Rx (see \eqref{M-MAC channel2}), as illustrated in Fig.~\ref{fig:rotation constellation}.
	If $\theta = 0$, the original decision region for $X_1$ is the optimal one, and $X_2$ has no interference in detecting $X_1$. If $\theta = \pi/4$, the received constellation is rotated, and hence, the original decision region is no longer the optimal detection region for $X_1$, and $X_2$ has interference in detecting $X_1$.
	
	\item Detection of $X_2$. The Rx first removes the additive interference based on $\hat{\bm X}_1$ assuming a successful detection, i.e., $\hat{\bm X}_1 = {\bm X}_1$, and then uses the MRC, which is discussed in Sec.~\ref{sec:R2}, and the power-detection based noncoherent detection method to detect $X_2$ since the Tag adopts a noncoherent modulation method.

\end{itemize}

\begin{figure}[t]
	\renewcommand{\captionfont}{\small} \renewcommand{\captionlabelfont}{\small}	
	\centering
	\begin{tikzpicture}[scale = 0.5]
	
\draw [ultra thick,->] (-5,0) -- (5,0);
\draw [ultra thick,->] (0,-4.5) --(0,5);
\draw (-4,-4)--(4,4);
\draw (4,-4)--(-4,4);

\node [fill,circle,blue,scale=0.5] (v10) at (2.5,2.5) {};
\node [fill,rectangle,red,scale=0.7] (v14) at (3.5,3.5) {};
\node [fill,circle,blue,scale=0.5] (v9) at (-2.5,2.5) {};
\node [fill,rectangle,red,scale=0.7] (v13) at (-3.5,3.5) {};
\node [fill,circle,blue,scale=0.5] (v12) at (-2.5,-2.5) {};
\node [fill,rectangle,red,scale=0.7] (v16) at (-3.5,-3.5) {};
\node [fill,circle,blue,scale=0.5] (v11) at (2.5,-2.5) {};
\node [fill,rectangle,red,scale=0.7] (v15) at (3.5,-3.5) {};

\node at (4.5,-0.5) {Real};
\node at (1.2,4.5) {Imag};

\draw [ultra thick,->]  (8,0) -- (18,0);
\draw[ultra thick,->]  (13,-4.5) --(13,5);
\draw (9,-4)--(17,4);
\draw (17,-4)--(9,4);

\node [fill,circle,blue,scale=0.5] (v2) at (15.5,2.5) {};
\node [fill,rectangle,red,scale=0.7] (v6) at (15.5,3.7) {};
\node [fill,circle,blue,scale=0.5] (v1) at (10.5,2.5) {};
\node [fill,rectangle,red,scale=0.7] (v5) at (9.3,2.5) {};
\node [fill,circle,blue,scale=0.5] (v4) at (10.5,-2.5) {};
\node [fill,rectangle,red,scale=0.7] (v8) at (10.5,-3.7) {};
\node [fill,circle,blue,scale=0.5] (v3) at (15.5,-2.5) {};
\node [fill,rectangle,red,scale=0.7] (v7) at (16.7,-2.5) {};
\draw (v1) -- (v2) -- (v3) -- (v4) -- (v1);
\draw [dashed] (v5) -- (v6) -- (v7) -- (v8) -- (v5);
\draw (v9) -- (v10) -- (v11) -- (v12) -- (v9);
\draw [dashed] (v13) -- (v14) -- (v15) -- (v16) -- (v13);

\node at (17.5,-0.5) {Real};
\node at (14.2,4.5) {Imag};
\node at (0,-5.5) {(a) $\theta = 0$};
\node at (13,-5.5) {(b) $\theta = \pi/4$};
	\end{tikzpicture}
	\vspace{-0.6cm}
	\caption{Illustration of the received signal constellation per Tx symbol duration without noise. The Tx and the Tag adopt the QPSK and on/off (i.e., $0/1$) modulation schemes, respectively. The circled (blue) and squared (red) dots denote for the received signal constellation when $X_2 = 0$ and $1$, respectively.}
	\vspace{-0.5cm}	
	\label{fig:rotation constellation}
\end{figure}
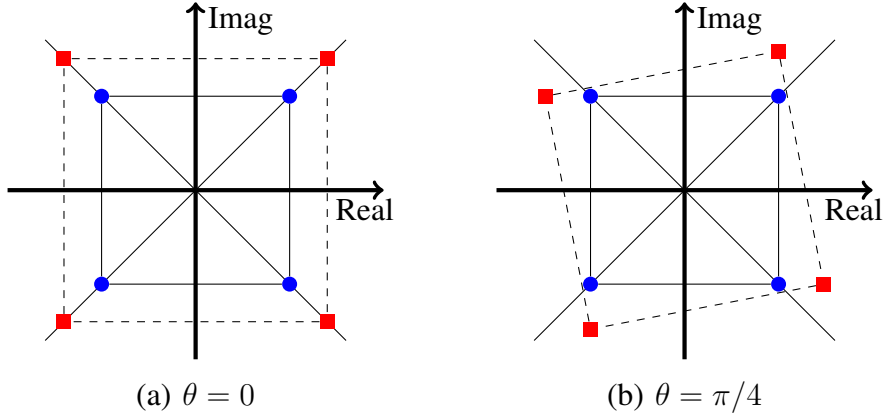

In the following, we analyze the detection error rates of $X_1$ and $X_2$ in the synchronous and asynchronous transmission scenarios.
For brevity, we assume that the Tx and the Tag adopt certain modulation schemes, i.e., the QPSK and on/off modulation schemes, respectively. 
Nevertheless, the analysis frame work is general and applicable for any $M$-ary modulated $X_1$ and noncoherently modulated $X_2$. Note that the on/off modulation scheme is commonly adopted in the standardized backscatter tags due to its simplicity and provision of a larger energy harvesting rate compared with the BPSK scheme in general.
The performance metrics of the detection error rates of $X_1$ and $X_2$ are the symbol-error rate (SER) and the bit-error rate (BER), respectively.  

\subsection{Error Rate Analysis: Synchronous Transmissions} \label{sec:syn}
Assuming that the Tx and the Tag are perfectly synchronized, we have the following results.
\subsubsection{Symbol Error Rate of $X_1$}
The SER of $X_1$ can be written as 
\begin{equation}\label{SER_X1_general}
\begin{aligned}
\myprobability{X_1 \neq \hat{X}_1} 
&= \myexpect{X_2}{\myprobability{X_1 \neq \hat{X}_1 \vert X_2}} 
= \frac{1}{2} \left(\myprobability{X_1 \neq \hat{X}_1 \vert X_2 = 0} + \myprobability{X_1 \neq \hat{X}_1 \vert X_2 = 1}\right).
\end{aligned}
\end{equation}
When $X_2 = 0$, the SER of $X_1$ is the same with that of the conventional QPSK modulation scheme, i.e.,
\begin{equation} \label{SER_X1_conditional_0}
\myprobability{X_1 \neq \hat{X}_1 \vert X_2 = 0} = 2 Q\left(\sqrt{\snr}\right) - Q^2\left(\sqrt{\snr}\right),
\end{equation}
When $X_2 = 1$, the SER of $X_1$ depends on the distribution of $\Theta$, and we have
\begin{align}
&\myprobability{X_1 \neq \hat{X}_1 \vert X_2 = 1} 
= \myexpect{\Theta}{\left.\myprobability{X_1 \neq \hat{X}_1 \vert X_2 = 1} \right\rvert \Theta}\\
\label{region}
&= \myexpect{\Theta}{
	Q\left(\sqrt{2 \snr }\left(\frac{1}{\sqrt{2}}+\vert g \vert \sqrt{\rho} \cos(\Theta + \frac{\pi}{4}) \right)\right) +Q\left(\sqrt{2 \snr }\left(\frac{1}{\sqrt{2}}+\vert g \vert \sqrt{\rho} \sin(\Theta + \frac{\pi}{4}) \right)\right)\right. \notag \\
&\hspace{1cm}\left.-Q\left(\sqrt{2 \snr }\left(\frac{1}{\sqrt{2}}+\vert g \vert \sqrt{\rho} \cos(\Theta + \frac{\pi}{4}) \right)\right)
Q\left(\sqrt{2 \snr }\left(\frac{1}{\sqrt{2}}+\vert g \vert \sqrt{\rho} \sin(\Theta + \frac{\pi}{4}) \right)\right)
}
\end{align}
\begin{align}
&= \frac{1}{\pi } 
\int_{0}^{2\pi}Q\left(\sqrt{2 \snr }\left(\frac{1}{\sqrt{2}}+\vert g \vert \sqrt{\rho} \cos(\theta) \right)\right) \mathrm d \theta \notag\\
&-
\frac{1}{2 \pi }
\int_{0}^{2\pi} Q\left(\!\sqrt{2 \snr }\left(\!\frac{1}{\sqrt{2}}+\vert g \vert \sqrt{\rho} \cos(\theta) \!\right)\!\right) 
Q\left(\!\sqrt{2 \snr }\left(\!\frac{1}{\sqrt{2}}+\vert g \vert \sqrt{\rho} \sin(\theta) \!\right)\!\right)
\mathrm d \theta\\
\label{SER_X1_conditional}
&\triangleq \mathcal{M}\left(\snr,\vert g \vert \sqrt{\rho}\right),
\end{align}
where \eqref{region} is due to the position of the received constellation points and the original decision regions [e.g., Fig.~\ref{fig:rotation constellation}(b)].
Then, taking \eqref{SER_X1_conditional_0} and \eqref{SER_X1_conditional} into \eqref{SER_X1_general}, we can obtain the following
result.
\begin{proposition} \label{SER_accurate}
	\normalfont
The SER of $X_1$ is given by
\begin{equation} \label{SER_accurate2}
\myprobability{X_1 \neq \hat{X}_1} 
= \frac{1}{2} \left(
2 Q\left(\sqrt{\snr}\right) - Q^2\left(\sqrt{\snr}\right)
+
\mathcal{M}\left(\snr,\vert g \vert \sqrt{\rho}\right)
\right).
\end{equation}
\end{proposition}

As there is no closed-form SER above, based on \eqref{SER_X1_conditional}, we further derive the easy-to-compute upper and lower bounds of \eqref{SER_accurate2} using the inequalities that $-1\leq \cos(\theta + \pi/4), \sin(\theta + \pi/4)\leq 1$, and we have
\begin{equation} \label{first_SER_UB}
\myprobability{X_1 \!\neq\! \hat{X}_1 \vert\! X_2 \!=\! 1} \!\!< \!
\overline{\mathcal{M}}\left(\snr,\vert g \vert \sqrt{\rho}\right)
\!\triangleq\!
2 Q\!\left(\!\!\sqrt{2 \snr }\left(\!\frac{1}{\sqrt{2}}\!-\!\vert g \vert \sqrt{\rho} \!\right)\!\right)
\!-Q^2\!\left(\!\!\sqrt{2 \snr }\left(\!\frac{1}{\sqrt{2}}\!+\!\vert g \vert \sqrt{\rho} \!\right)\!\right)\!.
\end{equation}

Moreover, the conditional SER in \eqref{SER_X1_conditional} is minimized when $\Theta = 0$, since the detection region of $X_1$ is correct and the effective SNR for $X_1$ is maximized. Therefore, we have
\begin{equation} \label{first_SER_LB}
\myprobability{\!X_1 \!\neq\! \hat{X}_1 \vert\! X_2 \!=\! 1} \!\!>\! 
\underline{\mathcal{M}}\left(\snr,\vert g \vert \sqrt{\rho}\right)
\!\triangleq\!
2 Q\!\left(\!\!\sqrt{2 \snr }\left(\!\frac{1}{\sqrt{2}}\!+\!\vert g \vert \sqrt{\rho}\! \right)\!\right) \!-\! Q^2\left(\!\!\sqrt{2 \snr }\left(\!\frac{1}{\sqrt{2}}\!+\!\vert g \vert \sqrt{\rho} \!\right)\!\right)\!.\!
\end{equation}

Therefore, taking \eqref{first_SER_UB} and \eqref{first_SER_LB} into \eqref{SER_X1_general}, we have the following results.
\begin{proposition} \label{X1_Bound}
	\normalfont	
	The upper and lower bounds of the SER of $X_1$, i.e., $\PUB$ and $\PLB$, are given by, respectively, 
	\begin{equation} \label{prop:lb}
	\begin{aligned}
		& P^{\text{UB}}_1 
		& \triangleq Q\left(\sqrt{\snr}\right)- \frac{1}{2}Q^2\left(\sqrt{\snr}\right) + \overline{\mathcal{M}}\left(\snr,\vert g \vert \sqrt{\rho}\right),\\
		& P^{\text{LB}}_1 
& \triangleq Q\left(\sqrt{\snr}\right)- \frac{1}{2}Q^2\left(\sqrt{\snr}\right) + \underline{\mathcal{M}}\left(\snr,\vert g \vert \sqrt{\rho}\right).		
	\end{aligned}
	\end{equation}
\end{proposition}
\begin{remark}
Both the upper and lower bounds are tight in the typical case that $\vert g \vert^2 \ll 1$, and the SER converges to the conventional one without backscatter transmissions.
\end{remark}

From Proposition~\ref{X1_Bound}, we have the following asymptotic upper and lower bounds in the high SNR scenario using the property that $Q(a x) \ll Q(b x)$ when $x \gg 1$ and $0<b<a$.
\begin{corollary} \label{SER_approx}
	\normalfont	
	In the high SNR scenario, the SER of $X_1$ is approximately bounded by
	\begin{equation}
	Q\left(\sqrt{\snr}\right)
	\leq
	\myprobability{X_1 \neq \hat{X}_1}
	\leq
	Q\left(\sqrt{2 \snr }\left(\frac{1}{\sqrt{2}}-\vert g \vert \sqrt{\rho} \right)\right).
	\end{equation}
\end{corollary}
\begin{remark}
The SER of $X_1$ in the multiplicative multiple-access system has almost the same decay rate in terms of the $\snr$ with the one without the Tag's interference, when $\vert g \vert^2 \ll 1$.
\end{remark}

\subsubsection{Bit Error Rate of $X_2$} \label{sec:BER1}
The BER of $X_2$ can be written as
\begin{equation}\label{BER1}
\begin{aligned}
\myprobability{X_2 \neq \hat{X}_2} 
&\!=\! \myprobability{\!X_2 \neq \hat{X}_2, \bm X_1 \!=\! \hat{\bm X}_1} \!+\! \myprobability{X_2 \neq  \hat{X}_2 \bm X_1 \neq \hat{\bm X}_1}.
\end{aligned}
\end{equation}
We calculate the first and second terms of the right side of \eqref{BER1} as follows:

\ul{If $\bm X_1$ is successfully detected}, after canceling the additive interference, performing MRC and scaling by the factor $\sqrt{\myP} \vert g \vert \sqrt{\rho} \vert \bm X_1 \vert \sigma$, the equivalent received signal, $\tilde{Y}$, is given by
\begin{equation} \label{tilde_Y_1}
\begin{aligned}
\tilde{Y}
&= e^{j \Theta} \sqrt{\snr} \vert g \vert \sqrt{\rho} \vert \bm X_1 \vert X_2
+ \frac{1}{\vert \bm X_1\vert \sigma} \sum_{i=1}^{N} X^*_{1,i} Z_i
= e^{j \Theta} \sqrt{\snr} \vert g \vert \sqrt{\rho} \vert \bm X_1 \vert X_2
+ \tilde{Z},
\end{aligned}
\end{equation}
where ${Z}_i \sim \mathcal{CN}\left(0,\sigma^2\right)$ and $\tilde{Z} \sim \mathcal{CN}\left(0,1\right)$.
Given $\Theta = \theta$, we further have
\begin{equation}
\tilde{Y}= \left\lbrace 
\begin{aligned}
&\tilde{Z} \sim \mathcal{CN}\left(0, 1 \right),&&X_2=0\\
&e^{j \theta} \sqrt{\snr} \vert g \vert \sqrt{\rho} \vert \bm X_1 \vert
+ \tilde{Z}
\sim \mathcal{CN}\left(e^{j \theta} \sqrt{\snr} \vert g \vert \sqrt{\rho} \vert \bm X_1 \vert, 1\right)
,&&X_2=1
\end{aligned}
\right.
\end{equation}
The optimal detector is power/energy detection based since $\theta$ is the unknown parameter~\cite{Shamai}, and we further define
\begin{equation}\label{chi-chi2}
\Xi \triangleq 2 \vert \tilde{Y} \vert^2
\sim
\left\lbrace 
\begin{aligned}
&\chi^2\left(2\right),&&X_2=0\\
& \text{noncentral } \chi^2\left(2, \lambda\right)
,&&X_2=1
\end{aligned}
\right.
\end{equation}
where 
\vspace{-0.3cm}
\begin{equation}\label{chi-chi2-lambda}
\lambda \triangleq 2 \snr \vert g \vert^2 \rho \vert \bm X_1\vert^2
=2 \snr \vert g \vert^2 \rho N.
\end{equation}
The pdfs of the $\chi^2(2)$ and $\text{noncentral } \chi^2 \left(2,\lambda \right)$ distribution are given by as, respectively,
\begin{align}
f_{\chi^2}\left(x,2\right) & = \frac{1}{2}e^{-\frac{x}{2}},\ 
f_{\text{nc-}\chi^2}\left(x,2,\lambda\right)  = \frac{1}{2} e^{-(x+\lambda)/2}I_0\left(\sqrt{\lambda x}\right),
\end{align}
where $I_\nu(\cdot)$ is the modified Bessel function of the first kind.

Assume that $\Lambda(\lambda)$ is the turning point of the sign of  $\left(f_{\chi^2}\left(x,2\right)-f_{\text{nc-}\chi^2}\left(x,2,\lambda\right)\right)$ as illustrated in Fig.~\ref{fig:chi2},~i.e., $f_{\chi^2}\left(x,2\right) >  f_{\text{nc-}\chi^2}\left(x,2,\lambda\right)$ when $0 \leq x < \Lambda(\lambda)$, and $f_{\chi^2}\left(x,2\right) <  f_{\text{nc-}\chi^2}\left(x,2,\lambda\right)$ when $x > \Lambda(\lambda)$.
Although the function $\Lambda(\lambda)$ does not have closed-form expression, $\Lambda(\lambda)$ can be easily evaluated using numerical methods, such as the bisection method, when $\lambda$ is given. 

\begin{figure*}[t]
	\renewcommand{\captionfont}{\small} \renewcommand{\captionlabelfont}{\small}
	\minipage{0.5\textwidth}
	\centering
	\includegraphics[scale=0.55]{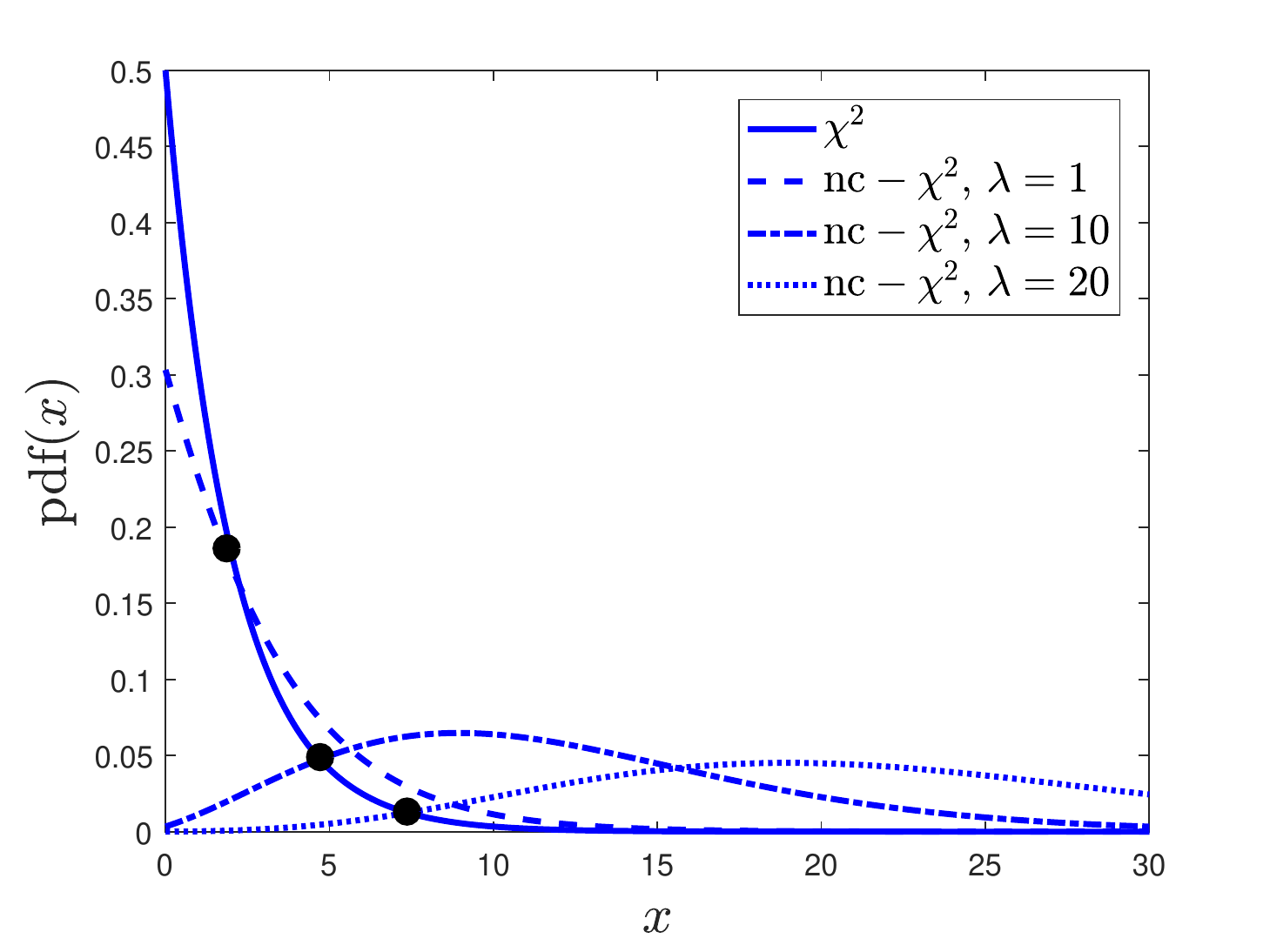}
\vspace{-0.6cm}
\caption{Illustration of the $\Lambda(\lambda)$, e.g., $\Lambda(1) = 2.25$, $\Lambda(10) = 4.71$ and $\Lambda(20) = 7.38$.}
\label{fig:chi2}
	\endminipage
	\hspace{0.3cm}
	\minipage{0.5\textwidth}
	\centering
	\includegraphics[scale=0.55]{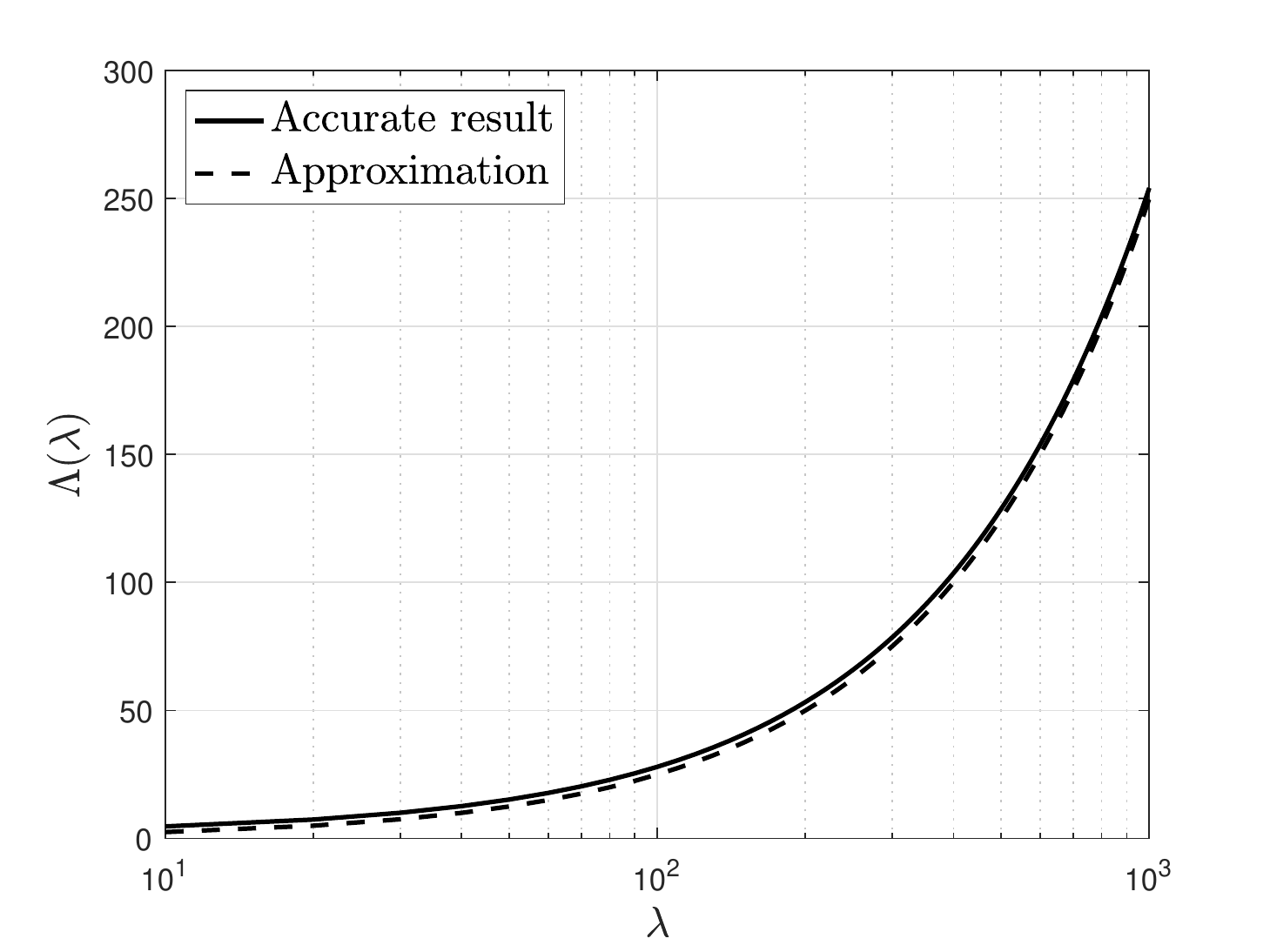}
\vspace{-0.6cm}
\caption{Verification of the approximation of the function $\Lambda(\lambda)$.}
\label{fig:verify_nice}
	\endminipage
	\vspace*{-0.0cm}
\end{figure*}

In the high SNR scenario, we have the following approximation.
\begin{corollary} \label{prop:nice}
	\normalfont
	In the high SNR scenario, the optimal detection threshold is $
	\Lambda\left(\lambda\right) = \frac{\lambda}{4}.
	$
\end{corollary}
\begin{proof}
	See Appendix~F.
\end{proof}

The verification of Corollary~\ref{prop:nice} is shown in Fig.~\ref{fig:verify_nice}, we see that the approximation is very tight when $\lambda > 10$.

Denote $\mathcal{H}_0$ and $\mathcal{H}_1$ as the hypotheses that $X_2=0$ and $1$, respectively. 
Thus, the optimal detector can be expressed as
\vspace{-0.3cm}
\begin{equation} \label{detection rule}
\Xi
 \mathrel{\mathop{ \lessgtr}^{\mathrm{\mathcal{H}_0}}_{\mathrm{\mathcal{H}_1}}} \Lambda\left(\lambda\right).
\end{equation}
\begin{remark}
	It can be observed from \eqref{chi-chi2} and \eqref{detection rule} that the detection error probability of $X_2$ directly depends on $\lambda$, and a larger $\lambda$ leads to a lower BER. Thus, $\lambda$ can be treated as the equivalent SNR for detecting $X_2$ at the Rx when $\bm X_1$ is successfully detected.
\end{remark}

Using the cdfs of the $\chi^2(2)$ and $\text{noncentral } \chi^2 \left(2,\lambda \right)$ distributions, i.e.,
\begin{align}\label{chi-cdf}
F_{\chi^2}\left(x,2\right) & = 1- e^{-\frac{x}{2}},\ 
F_{\text{nc-}\chi^2}\left(x,2,\lambda\right)  = 1-Q_1\left(\sqrt{\lambda},\sqrt{x}\right),
\end{align}
where $Q_M(a,b)$ is the Marum Q-function,
the error rate can be calculated as 
\begin{equation} 
\begin{aligned}
&\myprobability{X_2 \neq \hat{X}_2, \bm X_1 = \hat{\bm X}_1}
= \myprobability{X_2 \neq \hat{X}_2, X_2=0, \bm X_1 = \hat{\bm X}_1} + \myprobability{X_2 \neq \hat{X}_2, X_2=1, \bm X_1 = \hat{\bm X}_1}\\
&\!\!=\!\myprobability{\!X_2 \!\neq\! \hat{X}_2 \vert X_2\!=\!0,\! \bm X_1 \!=\! \hat{\bm X}_1\!}\!\! \myprobability{\!X_2\!=\!0,\! \bm X_1 \!=\! \hat{\bm X}_1\!} 
\!\!+\!\!\myprobability{\!X_2 \!\neq\! \hat{X}_2 \vert X_2\!=\!1,\! \bm X_1 \!=\! \hat{\bm X}_1\!}\!\! \myprobability{\!X_2\!=\!1,\! \bm X_1 \!=\! \hat{\bm X}_1\!} \notag
\end{aligned}
\end{equation}
\begin{equation} \label{BER-suc}
\begin{aligned}
&\!\!= \frac{1}{2}
\left(
\myprobability{\left. \Xi > \Lambda(\lambda) \right\vert X_2 = 0} 
\myprobability{\bm X_1 \!=\! \hat{\bm X}_1 \vert \!X_2\!=\!0} 
+\myprobability{\left. \Xi < \Lambda(\lambda) \right\vert X_2 = 1}
\myprobability{\bm X_1 \!=\! \hat{\bm X}_1 \vert \!X_2\!=\!1} 
\right)\\
&\!\!=\frac{1}{2}
\left(
\exp\left(- \frac{\Lambda(\lambda)}{2}\right)
\myprobability{\bm X_1 \!=\! \hat{\bm X}_1 \vert \!X_2\!=\!0} 
+
\left(1 - Q_{1}\left(\sqrt{\lambda},\sqrt{\Lambda(\lambda)}\right)\right)
\myprobability{\bm X_1 \!=\! \hat{\bm X}_1 \vert \!X_2\!=\!1} 
\right).
\end{aligned}
\end{equation}

\ul{If $\bm X_1$ is not successfully detected}, the additive interference cannot be canceled completely, and the MRC cannot provide the SNR gain as expected. Due to the residue interference, the BER is much worse compared with the case of successful detection of $\bm X_1$. Thus, we have
\begin{equation}\label{inequal}
\myprobability{X_2 \neq  \hat{X}_2, \bm X_1 = \hat{\bm X}_1}
<\myprobability{X_2 \neq  \hat{X}_2 \vert \bm X_1 = \hat{\bm X}_1}
<
\myprobability{X_2 \neq  \hat{X}_2 \vert \bm X_1 \neq \hat{\bm X}_1} < \frac{1}{2}.
\end{equation}
Moreover, the conditional sequence detection success rates of $X_1$ can be derived as 
\begin{align}
\myprobability{\bm X_1 = \hat{\bm X}_1 \vert X_2 = 0}
&= \left(1 - 2 Q\left(\sqrt{\snr}\right) + Q^2\left(\sqrt{\snr}\right)\right)^N,
\end{align}
\begin{align}\label{FER}
\myprobability{\bm X_1 = \hat{\bm X}_1 \vert X_2 = 1}
&= \myexpect{\Theta}{ \left.\myprobability{ \left. X_{1,i} = \hat{X}_{1,i}: \forall i \in \lbrace 1,2,...,N  \rbrace \right\vert X_2 = 1 } \right\rvert \Theta}\\ 
&=\myexpect{\Theta}{\left. \myprobability{\left. X_1= \hat{X}_1 \right\vert X_2 = 1}^N \right\rvert \Theta}.
\end{align}
Based on Proposition~\ref{X1_Bound} and the discussion above it, we can obtain 
\begin{equation} \label{X1_sequence_bound}
\left(1-\overline{\mathcal{M}}\left(\snr,\vert g \vert \sqrt{\rho}\right)\right)^N < \myprobability{\bm X_1 = \hat{\bm X}_1 \vert X_2 = 1} < \left(1-\underline{\mathcal{M}}\left(\snr,\vert g \vert \sqrt{\rho}\right)\right)^N.
\end{equation}

\ul{The BER of the Tag} is hence obtained by taking \eqref{inequal}, \eqref{X1_sequence_bound} and \eqref{BER-suc} into \eqref{BER1}, as follows:
\begin{proposition} \label{prop:ber_bounds}
	\normalfont
	The upper and lower bounds of the BER of $X_2$ are given by, respectively,
	\begin{equation} \label{prop:ber ub}
\begin{aligned}
	\PUBber 
	&=
	\frac{1}{2}
	\left(
	\exp\left(- \frac{\Lambda(\lambda)}{2}\right)
	+
	\left(1 - Q_{1}\left(\sqrt{\lambda},\sqrt{\Lambda(\lambda)}\right)\right)
	\right)
		\left(
	1- \underline{\mathcal{M}}\left(\snr,\vert g \vert \sqrt{\rho}\right)
	\right)^N	\\
	&+\frac{1}{2}\left(1 - \frac{1}{2} \left(\left(1 - 2 Q\left(\sqrt{\snr}\right) + Q^2\left(\sqrt{\snr}\right)\right)^N
	+
	\left(1-\overline{\mathcal{M}}\left(\snr,\vert g \vert \sqrt{\rho}\right)\right)^N
	\right)\right).
\end{aligned}
	\end{equation}
and	
\begin{equation} \label{prop:ber lb}
\hspace{-5cm}
\begin{aligned}
	\PLBber 
	&=
	\frac{1}{2}
	\left(
	\exp\left(- \frac{\Lambda(\lambda)}{2}\right) 
	\left(1 - 2 Q\left(\sqrt{\snr}\right) + Q^2\left(\sqrt{\snr}\right)\right)^N\right.\\
	&\left. + 
	\left(1 - Q_{1}\left(\sqrt{\lambda},\sqrt{\Lambda(\lambda)}\right)\right)
	\left(1-\overline{\mathcal{M}}\left(\snr,\vert g \vert \sqrt{\rho}\right)\right)^N 
	\right).
\end{aligned}
	\end{equation}	
\end{proposition}

Moreover, the asymptotic BER in the high SNR scenario can be obtained as follows.

\begin{corollary} \label{prop:asynBER}
	\normalfont
In the high SNR scenario, the gap between the upper and lower bounds in Proposition~\ref{prop:ber_bounds} converges to zero in the typical case that $\vert g \vert^2 \ll 1$, and the asymptotic BER is given by
\begin{equation}\label{BER_approx}
\myprobability{X_2 \neq \hat{X}_2} = \frac{1}{2} \left(\exp\left(-\frac{\lambda}{8}\right) +
\left(1-Q_1\left(\sqrt{\lambda}, \frac{\sqrt{\lambda}}{2}\right)\right)
\right).
\end{equation}
\end{corollary}
\begin{proof}
	See Appendix~G.	
\end{proof}
From Corollary~\ref{prop:asynBER} and \eqref{chi-cdf}, we have the following remark.
\begin{remark}
	The BER of $X_2$ in the multiplicative multiple-access system has  the same decay rate in terms of the $\snr$ with the one without the Tx's interference, when $\vert g \vert^2 \ll 1$.
\end{remark}

\subsection{Error Rate Analysis: Asynchronous Transmissions}
Assume that the delay offset between the Tx and the Tag of the asynchronous scenario is $\alpha T_1$ as illustrated in Fig.~\ref{fig:asyn}, where
$\alpha$ is the delay offset ratio.
Without loss of generality, we assume that $\alpha \in \left[0,0.5\right]$.
Without increasing the implementation complexity of the Rx, we assume that the Rx does not perform oversampling during the information detection, and the value of $\alpha$ is not known at the Rx.
Recall that the sampling period of the Rx is the same with the symbol duration of the Tx.
For instance, in Fig.~\ref{fig:asyn}, the Rx uses the first three samplings (thick framed) to detect $X_2$.
In the following, we analyze the detection error rates under the asynchronous setting using the detection method discussed in Sec.~\ref{sec:syn}.

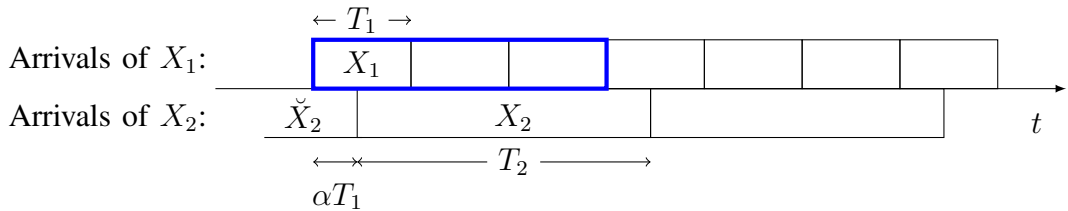
\begin{figure}[t]
	\renewcommand{\captionfont}{\small} \renewcommand{\captionlabelfont}{\small}	
	\centering
	\begin{tikzpicture}[scale=1.3]
\draw [-latex] (-1,0) -- (7.7,0);
\draw  (0,0.5) node (v3) {} rectangle (1,0);
\draw  (1,0.5) rectangle (2,0);
\draw  (2,0.5) rectangle (3,0);
\draw  (3,0.5) rectangle (4,0);
\draw  (4,0.5) rectangle (5,0);
\draw  (0.45,0) rectangle (3.45,-0.5);
\draw  (3.45,0) rectangle (6.45,-0.5);
\draw  (5,0.5) rectangle (6,0);
\draw  (6,0.5) rectangle (7,0);
\node at (7.4,-0.35) {$t$};
\node at (0.5,0.25) {$X_1$};
\node at (2.05,-0.3) {$X_2$};

\node (v1) at (0.5,0.7) {$T_1$};
\draw [->] (v1) -- (0,0.7);
\draw [->](v1) -- (1,0.7);
\node (v2) at (2.05,-0.75) {$T_2$};
\draw [->] (v2) -- (0.45,-0.75);
\draw [->] (v2) -- (3.45,-0.75);
\draw [<->] (0,-0.75) -- (0.45,-0.75);
\node at (0.25,-1.1) {$\alpha T_1$};
\node at (-2.1,0.3) {Arrivals of $X_1$:  };
\node at (-2.1,-0.3) {Arrivals of $X_2$:  };
\node at (-0.1,-0.27) {$\breve{X}_2$};
\draw [ultra thick,blue] (v3) rectangle (3,0);
\draw (0.5,-0.5) -- (-0.5,-0.5);
	\end{tikzpicture}
\vspace{-0.6cm}	
	\caption{Illustration of asynchronous transmissions, $N=3$.}
	\label{fig:asyn}
\vspace{-0.8cm}	
\end{figure}

\subsubsection{Symbol Error Rate of $X_1$}
As illustrated in Fig.~\ref{fig:asyn}, there is one symbol of every $N$ Tx symbols that is different from the others for detection, i.e., the detection is affected by two $X_2$ symbols. The detection of each of the rest $N-1$ symbols is affected by only one $X_2$ symbol, i.e.,  the detection error rate is the same with the synchronous transmission scenario. 
In the following, we thus focus on the detection error  rate of the $X_1$ symbol that is affected by two symbols of the Tag denoted by $\breve{X}_2$ and $X_2$, respectively.
The received signal expression is rewritten as
\begin{equation}
Y= X_1 \left(1 + g \sqrt{\rho} \left(\alpha \breve{X}_2 + (1-\alpha) X_2\right) \right) +Z.
\end{equation}
Thus, the detection error rate of the symbol is given by
\begin{equation} \label{asyn_X_1_1}
\begin{aligned}
&{\myprobability{X_1 \neq \hat{X}_1 \vert \text{affected by both } \breve{X}_2 \text{ and } X_2}}\\
&= \frac{1}{4}
\left(
\myprobability{X_1 \neq \hat{X}_1 \vert \breve{X}_2=0, X_2=0}
+\myprobability{X_1 \neq \hat{X}_1 \vert \breve{X}_2=0, X_2=1}\right.\\
&\left. \hspace{1cm}+\myprobability{X_1 \neq \hat{X}_1 \vert \breve{X}_2=1, X_2=0}
+\myprobability{X_1 \neq \hat{X}_1 \vert \breve{X}_2=1, X_2=1}
\right),
\end{aligned}
\end{equation}
where $\myprobability{X_1 \neq \hat{X}_1 \vert \breve{X}_2=0, X_2=0}$ and $\myprobability{X_1 \neq \hat{X}_1 \vert \breve{X}_2=1, X_2=1}$ have been obtained in \eqref{SER_X1_conditional_0} and \eqref{SER_X1_conditional}, respectively, i.e., the detection error rates are the same with that of the synchronous scenario.
Based on the analysis in \eqref{SER_X1_conditional}, it can be derived that  
\begin{equation} \label{myij}
\myprobability{X_1 \neq \hat{X}_1 \vert \breve{X}_2=i, X_2=j} =
\left\lbrace
\begin{aligned}
\small
&\mathcal{M}\left(\snr, 0 \right), && i=0, j=0\\
&\mathcal{M}\left(\snr,(1-\alpha)\vert g \vert \sqrt{\rho}\right), && i=0, j=1 \\
&\mathcal{M}\left(\snr,\alpha\vert g \vert \sqrt{\rho}\right), && i=1, j=0 \\
&\mathcal{M}\left(\snr,\vert g \vert \sqrt{\rho}\right), && i=1, j=1 \\
\end{aligned}
\right.
\end{equation}
Therefore, the detection error rate in \eqref{asyn_X_1_1} is further derived as
\begin{equation}
\begin{aligned}
&{\myprobability{X_1 \neq \hat{X}_1 \vert \text{affected by } \breve{X}_2 \text{ and } X_2}}\\
&=\frac{1}{4} \left(
\mathcal{M}\left(\snr, 0\right) + \mathcal{M}\left(\snr,(1-\alpha)\vert g \vert \sqrt{\rho}\right)
+\mathcal{M}\left(\snr,\alpha\vert g \vert \sqrt{\rho}\right)
+\mathcal{M}\left(\snr,\vert g \vert \sqrt{\rho}\right)
\right).
\end{aligned}
\end{equation}
Since the SER of the asynchronous transmission scenario can be expressed as
\begin{equation}
\begin{aligned}
&\myprobability{X_1 \neq \hat{X}_1 \vert \alpha \neq 0} \\
&= 
\frac{1}{N} \myprobability{X_1 \neq \hat{X}_1 \vert \text{affected by both } \breve{X}_2 \text{ and } X_2} + \frac{N-1}{N} \myprobability{X_1 \neq \hat{X}_1 \vert \text{affected by } X_2 \text{ only}},
\end{aligned}
\end{equation}
we have the following result.
\begin{proposition} \label{prop:asyn_ser}
\normalfont	
	In the asynchronous transmission scenario, the SER of $X_1$ is 
	\begin{equation}
	\begin{aligned}
	P^{\text{Asyn}}_1 \triangleq \myprobability{X_1 \neq \hat{X}_1 \vert \alpha \neq 0} 
	&= \left(\frac{1}{4N} +\frac{N-1}{2N}\right)\left(\mathcal{M}\left(\snr,0\right) + \mathcal{M}\left(\snr,\vert g \vert \sqrt{\rho}\right)\right)\\
	&+ \frac{1}{4N} \left(\mathcal{M}\left(\snr,\alpha\vert g \vert \sqrt{\rho}\right) + \mathcal{M}\left(\snr,(1-\alpha)\vert g \vert \sqrt{\rho}\right)\right).
	\end{aligned}
	\end{equation}
\end{proposition}

Moreover, similar to the analysis of the synchronous transmission scenario, the asymptotic results are given by as follows.
\begin{proposition} \label{prop:asyn_ser_bound}
\normalfont
In the high SNR scenario, the asymptotic upper and lower bounds of the SER of $X_1$ are given by, respectively,
\begin{equation}\label{asyn_SER_high_snr}
\begin{aligned}
\PUBasyn &= \frac{2N-1}{2 N}Q\left(\sqrt{2 \snr }\left(\frac{1}{\sqrt{2}}-\vert g \vert \sqrt{\rho} \right)\right),\ 
\PLBasyn = \frac{2N-1}{2 N}Q\left(\sqrt{\snr }\right).
\end{aligned}
\end{equation}
\end{proposition}
\begin{remark}
	Comparing Proposition~\ref{prop:asyn_ser_bound} with Corollary~\ref{SER_approx}, we see that the asynchronous transmission helps to reduce the SER of $X_1$.
\end{remark}

\subsubsection{Bit Error Rate of $X_2$}
In the scenario that the sequence of symbols, i.e., $\bm X_1$, has been successfully detected, after removing the additive interference, performing MRC and scaling by the factor $\sqrt{\myP} \vert g \vert \sqrt{\rho} \vert \bm X_1 \vert \sigma$, the equivalent received signal, $\tilde{Y}$, is given by
\begin{equation}\label{tilde_Y_2}
\begin{aligned}
\tilde{Y}
&= e^{j \Theta} \sqrt{\snr} \vert g \vert \sqrt{\rho} \left( 
\vert \bm X_1 \vert X_2 + \frac{\vert X_{1,1} \vert^2}{\vert \bm X_1 \vert } \alpha \left(\breve{X}_2 - X_2\right)
\right)
+ \tilde{Z}.
\end{aligned}
\end{equation}
It is easy to see that this signal expression degrades to the synchronous one, i.e., \eqref{tilde_Y_1}, when $\breve{X}_2 = X_2$. 
Thus, we focus on the scenarios that $\breve{X}_2 \neq X_2$ and the detection of $X_2$ suffers from the inter-symbol interference.

Then, we can derive the received signal expression
\begin{equation}
\tilde{Y}= \left\lbrace 
\begin{aligned}
&e^{j \Theta} \sqrt{\snr} \vert g \vert \sqrt{\rho} \frac{\vert X_{1,1} \vert^2}{\vert \bm X_1 \vert } \alpha +\tilde{Z},&&X_2=0,\ \breve{X}_2 = 1\\
&e^{j \Theta} \sqrt{\snr} \vert g \vert \sqrt{\rho} \left(\vert \bm X_1 \vert - \frac{\vert X_{1,1} \vert^2}{\vert \bm X_1 \vert } \alpha \right)
+ \tilde{Z}
,&&X_2=1,\ \breve{X}_2=0
\end{aligned}
\right.
\end{equation}
and the relevant distributions
\begin{equation}\label{phi}
\Phi \triangleq 2 \vert \tilde{Y} \vert^2
\sim
\left\lbrace 
\begin{aligned}
& \text{noncentral } \chi^2\left(2,\lambda_0\right),&&X_2=0,\ \breve{X}_2=1\\
& \text{noncentral } \chi^2\left(2, \lambda_1\right)
,&&X_2=1,\ \breve{X}_2=0
\end{aligned}
\right.
\end{equation}
where 
\begin{equation}\label{phi-lambda}
\begin{aligned}
\lambda_0 &\triangleq 2 \snr \vert g \vert^2 \rho \vert \bm X_1\vert^2
=2 \snr \vert g \vert^2 \rho \frac{1}{N} \alpha^2,\\
\lambda_1 &\triangleq 2 \snr \vert g \vert^2 \rho \vert \bm X_1\vert^2
=2 \snr \vert g \vert^2 \rho \left(\sqrt{N} - \frac{\alpha}{\sqrt{N}}\right)^2 =2 \snr \vert g \vert^2 \rho \left(N+\frac{\alpha^2}{N} -2 \alpha\right).
\end{aligned}
\end{equation}
Thus, the detection error rate can be calculated as
\begin{equation}
\begin{aligned}
&\myprobability{\hat{X}_2 \neq X_2, \hat{\bm X}_1 = \bm X_1} \\
&=\! \sum_{i=0}^{1} \sum_{j=0}^{1} \myprobability{\!\hat{X}_2 \neq X_2 \vert \breve{X}_2 \!=\! i, X_2\!=\!j, \hat{\bm X}_1 = \bm X_1\!}\! \myprobability{\!\hat{\bm X}_1 \!=\! \bm X_1 \vert \breve{X}_2 \!=\! i, X_2\!=\!j}
\!\myprobability{\!\breve{X}_2 \!=\! i, X_2\!=\!j}\\
& = \frac{1}{4}
 \sum_{i=0}^{1} \sum_{j=0}^{1} \myprobability{\hat{X}_2 \neq X_2 \vert \breve{X}_2 = i, X_2=j, \hat{\bm X}_1 = \bm X_1} \myprobability{\hat{\bm X}_1 = \bm X_1 \vert \breve{X}_2 = i, X_2=j},
\end{aligned}
\end{equation}
where $\myprobability{\hat{\bm X}_1 = \bm X_1 \vert \breve{X}_2 = i, X_2=j}$ can be obtained from Proposition~\ref{prop:asyn_ser}, and 
\begin{equation}
 \myprobability{\hat{X}_2 \neq X_2 \vert \breve{X}_2 = i, X_2=j, \hat{\bm X}_1 = \bm X_1}
 = 
 \left\lbrace 
 \begin{aligned}
 & \exp\left(- \frac{\Lambda(\lambda)}{2}\right), && i=0,j=0\\
 & 1-Q_1\left(\sqrt{\lambda_1},\sqrt{\Lambda(\lambda)}\right), && i=0,j=1\\
 & Q_1\left(\sqrt{\lambda_0},\sqrt{\Lambda(\lambda)}\right), && i=1,j=0\\  
 & \left(1 - Q_{1}\left(\sqrt{\lambda},\sqrt{\Lambda(\lambda)}\right)\right), && i=1,j=1
 \end{aligned}
 \right.
\end{equation}

Therefore, similar to the analysis of the synchronous transmission scenario, the upper and~lower bounds of the BER can be derived. Only the asymptotic results are presented below for brevity.
\begin{proposition} \label{prop:ber_bound_asyn}
	\normalfont
	In the high SNR scenario, the lower and upper bounds of the BER of the asynchronous transmission scenario are given by, respectively,
	\begin{equation}
	\begin{aligned}
		&\PUBberasyn=
		\frac{1}{4}\! \left(\!2 \!+ \exp\left(\!- \frac{\Lambda(\lambda)}{2}\!\right) + Q_1\left(\!\sqrt{\lambda_0},\sqrt{\Lambda(\lambda)}\!\right)
		-Q_1\left(\!\sqrt{\lambda},\sqrt{\Lambda(\lambda)}\!\right)
		-Q_1\left(\!\sqrt{\lambda_1},\sqrt{\Lambda(\lambda)}\!\right)
		\!\right) \\
		&\hspace{2cm} +\frac{1}{2}\left(1 - \left(1-\overline{\mathcal{M}}\left(\snr,\vert g \vert \sqrt{\rho}\right)\right)^N\right),\\
		&\PLBberasyn=
		\frac{1}{4}\! \left(\!2 \!+ \exp\left(\!- \frac{\Lambda(\lambda)}{2}\!\right) + Q_1\left(\!\sqrt{\lambda_0},\sqrt{\Lambda(\lambda)}\!\right)
\!-Q_1\left(\!\sqrt{\lambda},\sqrt{\Lambda(\lambda)}\!\right)
\!-Q_1\left(\!\sqrt{\lambda_1},\sqrt{\Lambda(\lambda)}\!\right)
\!\right)\!, \notag
		\end{aligned}
	\end{equation}	
	where $\Lambda(\lambda) = \frac{\lambda}{4}$, which is given in Proposition~\ref{prop:nice}. Moreover, the gap between the upper and lower bounds converges to zero.
\end{proposition}
\begin{remark}
Based on Proposition~\ref{prop:ber_bound_asyn}, it is easy to verify that the asymptotic BER increases with the delay offset ratio $\alpha$. 
\end{remark}

In addition, a simpler method for detecting $X_2$ is to ignore part of the signal that is affected by the previous symbol in the detection, and only use the samplings without inter-symbol interference. For example, as illustrated in Fig.~\ref{fig:asyn}, only the second and third samplings of $X_1$ are used for detecting $X_2$.
Assuming that $\bm X_1$ has been successfully detected, to detection $X_2$, the only difference with the synchronous transmission scenario in Sec.~\ref{sec:BER1} is that the equivalent $\snr$ for the detection shrinks with the ratio $(N-1)/N$, i.e., 
\begin{equation}\label{lambda2}
\lambda = 2 \snr \vert g \vert^2 \rho \left(\vert \bm X_1\vert^2 - \vert X_{1,1} \vert^2\right)
=2 \snr \vert g \vert^2 \rho (N-1),\ N \geq 2.
\end{equation}

\subsection{Numerical Results}
We present numerical results of the detection error rates of the multiplicative multiple-access system for both the synchronous and asynchronous transmission scenarios. The simulation results of the SER and the BER of $X_1$ and $X_2$ are plotted based on Monte Carlo simulation with $10^9$ points using the detection method in Sec~\ref{sec:detection method}.

\subsubsection{Detection Error Rates in the Synchronous Transmission Scenario}
In Fig.~\ref{fig:Syn_X1}, the analytical result of the SER of $X_1$ is plotted using Proposition~\ref{SER_accurate} for different $\snr$ and relative backscatter channel power gain, $\vert g \vert^2 \rho$, and the upper and lower bounds of the SER are plotted using Proposition~\ref{X1_Bound}.
We see that the analytical result matches well with the simulation result, which verifies the accuracy of the former, and the SER decreases with increasing $\snr$.
Also, we see that the gap between the upper and lower bounds decreases inversely with $\vert g \vert^2 \rho$, e.g., the ratio of the upper bound to the lower bound is $1.5\times 10^{5}$ and $40$ when $\vert g \vert^2 \rho = 0.1$ and $0.01$, respectively, when $\snr = 15$~dB. Moreover, the SER approaches to the lower bound with decreasing $\vert g \vert^2 \rho$, e.g., changing from $0.1$ to $0.01$, since the interference from the Tag's transmissions reduces. \emph{Therefore, the SER is very close to the lower bound when $\vert g \vert^2 \rho = 0.01$, i.e., the detection of the Tx's signal almost does not suffer from the Tag's interference, and the performance of the proposed detection method is closed to the optimal one achieved by the ML method.}

\begin{figure*}[t]
	\renewcommand{\captionfont}{\small} \renewcommand{\captionlabelfont}{\small}
	\minipage{0.5\textwidth}
	\centering
	\includegraphics[scale=0.6]{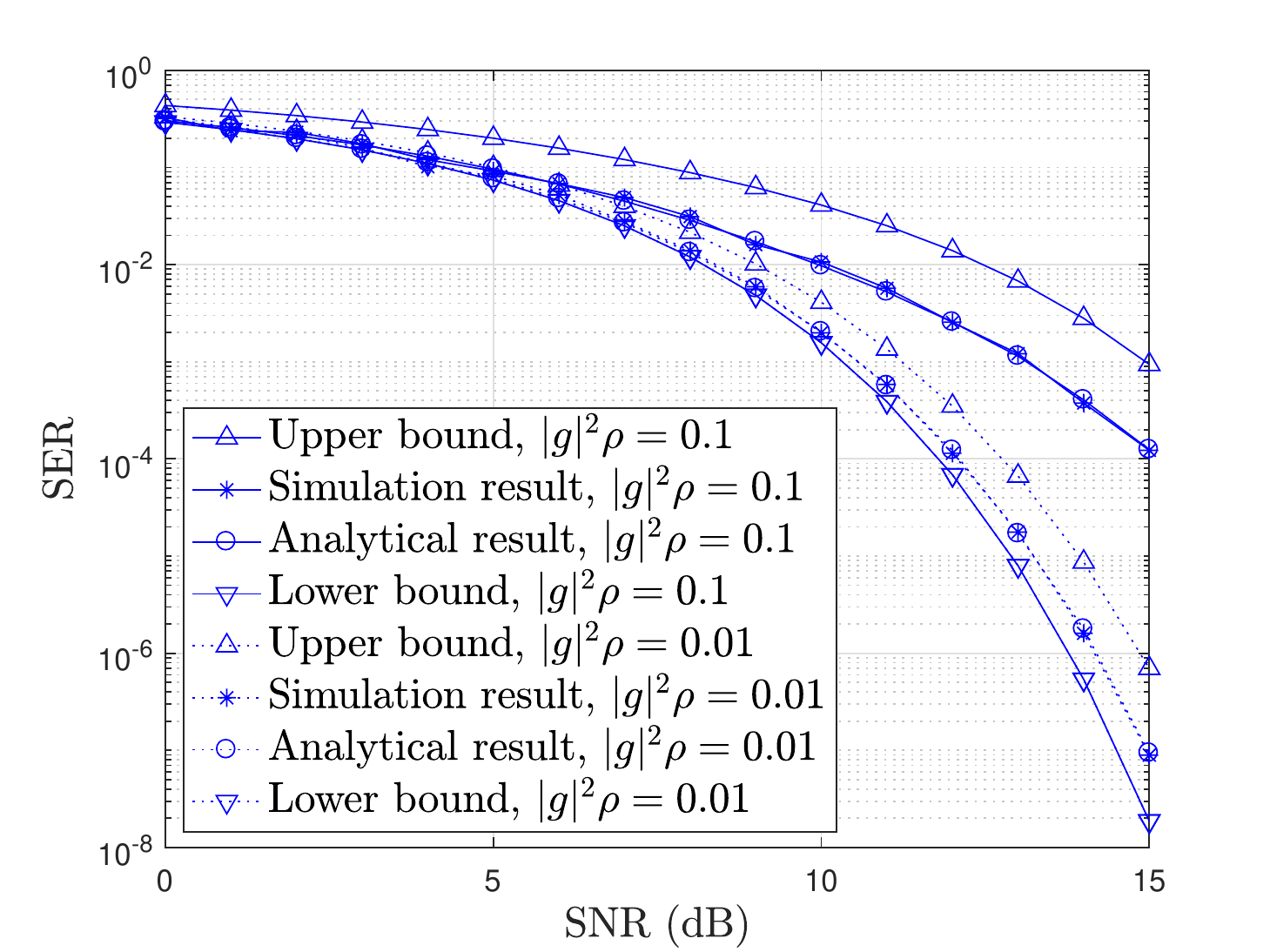}
	\vspace{-0.9cm}
	\caption{SER of $X_1$ versus SNR.}
		\vspace{-0.3cm}
	\label{fig:Syn_X1}
	\endminipage
	\minipage{0.5\textwidth}
	\centering
	\includegraphics[scale=0.6]{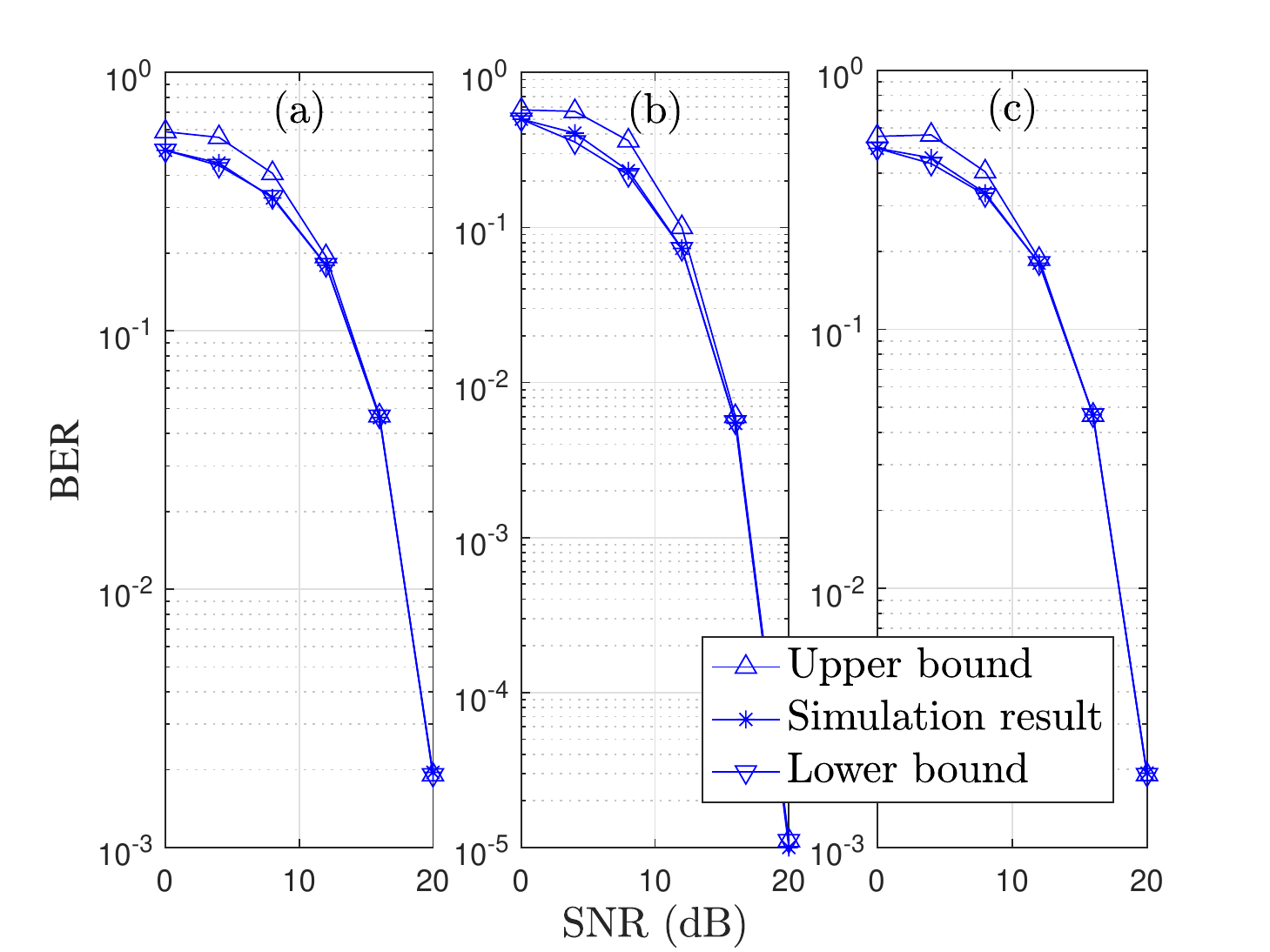}
		\vspace{-0.9cm}
	\caption{BER of $X_2$ versus SNR: (a) $N=2$, $\vert g \vert^2 \rho = 0.1$; (b) $N=4$, $\vert g \vert^2 \rho = 0.1$; and (c) $N=4$, $\vert g \vert^2 \rho = 0.05$.}
		\vspace{-0.5cm}
	\label{fig:Syn_X2}
	\endminipage
	\vspace*{-0.5cm}
\end{figure*}

In Fig.~\ref{fig:Syn_X2}, the simulation result of the BER of $X_2$ is plotted for different $\snr$, symbol-length ratio, $N$, and relative backscatter channel power gain, $\vert g \vert \rho$. 
The upper and lower bounds of the SER are plotted using Proposition~\ref{prop:ber_bounds}.
We see that the gap between the upper and lower bounds vanishes quickly with increasing $\snr$, and the bounds are very tight when $\snr > 10$~dB.
\emph{Therefore, the BER is very close to the lower bound when $\snr > 10$~dB, i.e., the detection of the Tag's signal almost does not suffer from the Tx's interference, and the performance of the proposed detection method is closed to the optimal one achieved by the ML method.}
Also, we see that the BER decreases with increasing $N$ (see Fig.~\ref{fig:Syn_X2}(a) and (b)), e.g., the BER reduces by a factor of $100$ when we increase $N$ from $2$ to $4$ with $\snr = 20$, and $\vert g \vert^2 \rho = 0.1$; we see that the BER reduces by a factor of $100$ when changing $\vert g \vert^2 \rho$ from $0.05$ to $0.1$ (see Fig.~\ref{fig:Syn_X2}(c) and (b)), because of the increasing equivalent SNR for the detection.

\subsubsection{Detection Error Rates in the Asynchronous Transmission Scenario}

In Fig.~\ref{fig:Asyn_X1}, the analytical result of the BER of $X_1$ is plotted using Proposition~\ref{prop:asyn_ser} for different delay offset ratio, $\alpha$, and symbol-length ratio, $N$. 
We see that the analytical result matches well with the simulation result, which verifies the accuracy of the former.
Also, we see that the SER decreases with increasing delay offset ratio and decreasing $N$. 
For instance, the SER decreases from $9.8 \times 10^{-3}$ to $8.4 \times 10^{-3}$ when $N=2$ and the delay offset ratio increases from $0$ to $0.5$, i.e., from the synchronous scenario to a very asynchronous scenario; the SER decreases from $9.35 \times 10^{-3}$ to $8.4 \times 10^{-3}$ when the relative symbol length decreases from $6$ to $2$ with $\alpha = 0.5$.
This is because the more asynchronization between the Tx and the Tag the lower probability that the detection of $X_1$ suffers from strong interference of $X_2$, i.e., the scenario that $X_2 = 1$.

\begin{figure*}[t]
	\renewcommand{\captionfont}{\small} \renewcommand{\captionlabelfont}{\small}
	\hspace{-0.1cm}	
	\minipage{0.5\textwidth}
	\centering
	\includegraphics[scale=0.6]{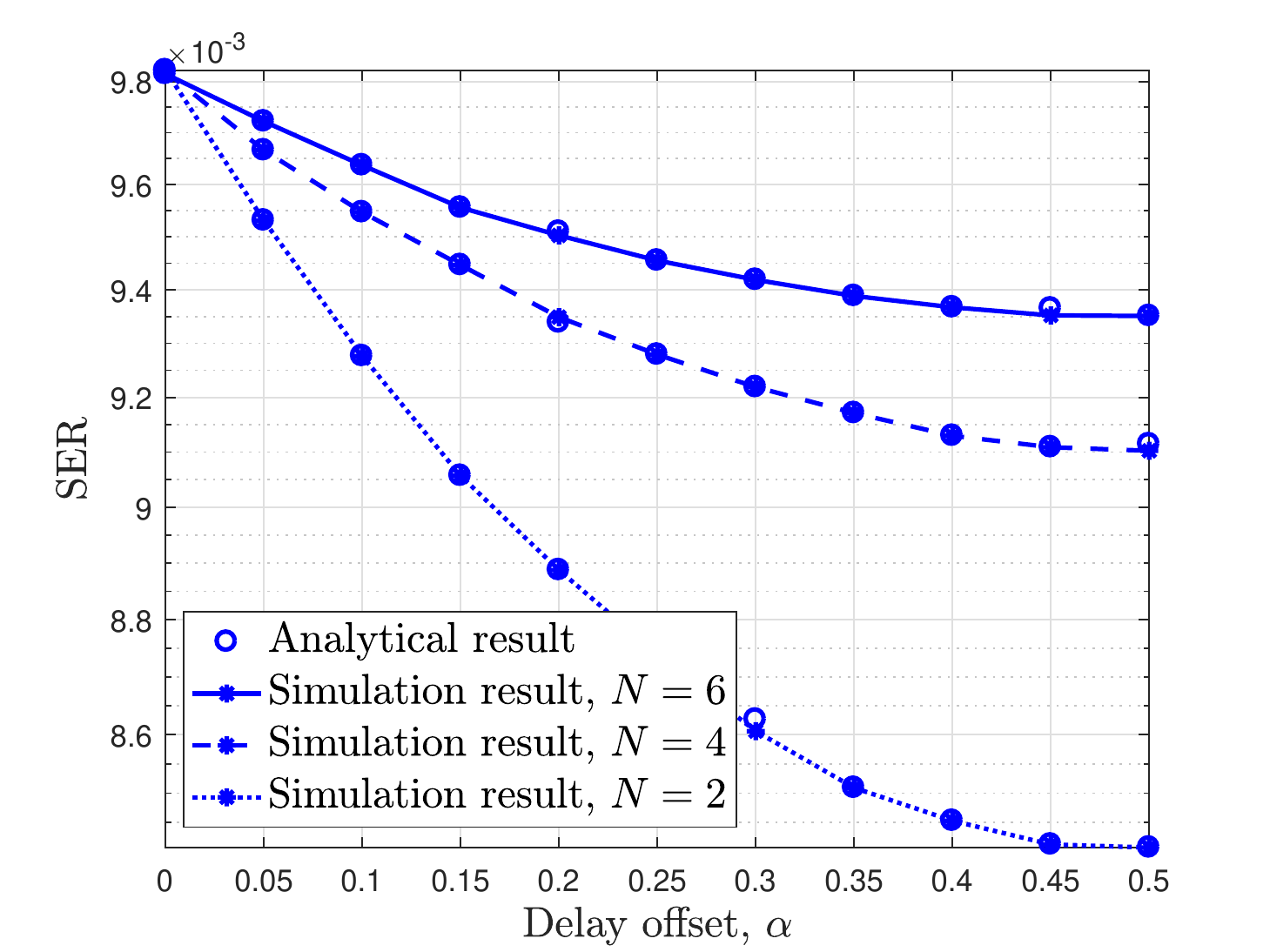}
	\vspace{-0.9cm}
\caption{The SER of $X_1$ versus the delay offset ratio $\alpha$, $\snr=10$~dB, $\vert g \vert^2 \rho = 0.1$.}
	\vspace{-0.5cm}
\label{fig:Asyn_X1}
	\endminipage
	\hspace{0.1cm}
	\minipage{0.5\textwidth}
	\centering
	\includegraphics[scale=0.6]{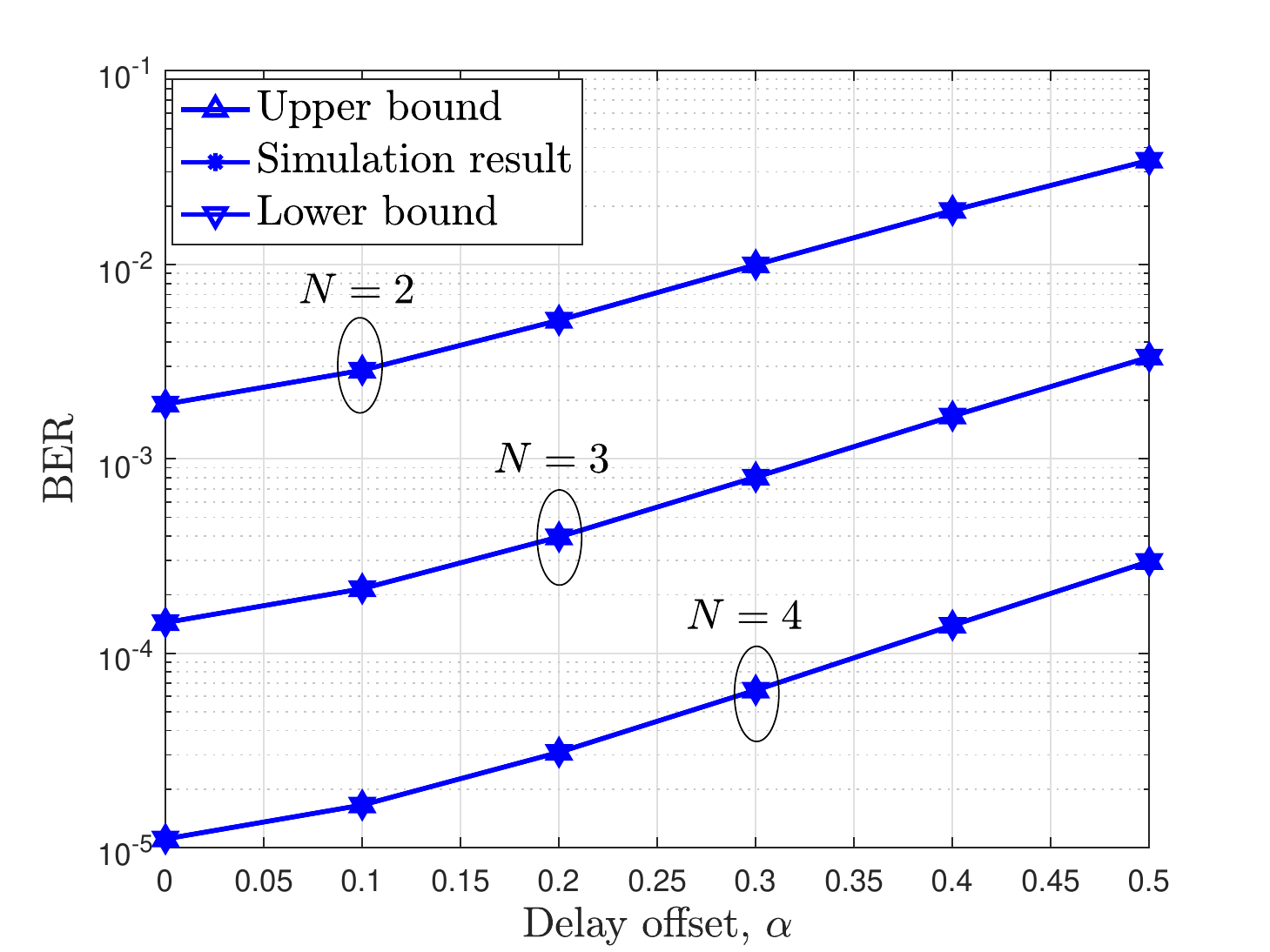}
		\vspace{-0.9cm}
\caption{The BER of $X_2$ versus the delay offset ratio $\alpha$, $\snr=20$~dB, $\vert g \vert^2 \rho = 0.1$.}
\label{fig:Asyn_X2}
	\vspace{-0.5cm}
	\endminipage
	\vspace*{-0.2cm}
\end{figure*}

In Fig.~\ref{fig:Asyn_X2}, 
the simulation result of the BER of $X_2$ is plotted for different delay offset ratio, $\alpha$, and symbol-length ratio, $N$. 
The upper and lower bounds of the SER are plotted using Proposition~\ref{prop:ber_bound_asyn}.
We see that both the upper and lower bounds are tight when $\snr = 20$~dB.
Moreover, we see that the BER performance is worse with a larger $\alpha$ or a smaller $N$.
This is because the more asynchronization between the Tx and the Tag the larger inter-symbol interference that the detection of $X_2$ suffers from.
Also, it is interesting for us to compare the performance between the two strategies, i.e., to use the entire signal, e.g., $N=3$ and $\alpha >0$, or simply use the part of the signal without inter-symbol interference, e.g., $N=2$ and $\alpha = 0$.
We see that the strategy that uses the entire signal for detection is better than the one that simply discards the part of the signal with inter-symbol interference, when the asynchronization between $X_1$ and $X_2$ is small or in a moderate range, i.e., $\alpha \in (0,0.4)$. However, the later is better than the former when the asynchronization is severe, i.e., $\alpha \in (0.4, 0.5)$.

\vspace{-0.2cm}
\section{Conclusions}
In this paper, we analyzed an ambient backscatter multiple-access system, which introduces the M-MAC.  
We proved that the achievable rate region of the M-MAC is strictly larger than that achieved by the conventional TDMA scheme in many cases, such as the high-SNR case and the typical case that the direct channel is much stronger than the backscatter channel.
The numerical results also showed that the proposed multiple-access scheme improves the rate performance of the system in the range of $0-30$~dB of the direct-link SNR.
Moreover, we comprehensively analyzed the detection error rates of a practical multiplicative multiple-access system in which the Tx and the Tag adopt practical coherent and noncoherent modulation schemes, respectively, in both the synchronous and asynchronous scenarios. Our results showed that the asynchronous transmission reduces the detection error rates of the Tx but increases that of the Tag.
For future work, we will consider the multiple-antenna Tx, Tag and Rx in the multiplicative multiple-access system.

\vspace{-0.3cm}
\section*{Appendix}
\subsection{Proof of Proposition~\ref{prop:time-sharing2}}
Based on the achievable rate analysis in Sec.~\ref{sec:R2}, by letting each element of $\bm X_1$ to be $1$, i.e., $\bm X_1$ is a deterministic/unmodulated signal, ${R}^{\max}_2$ is obtained by \eqref{IX2tildeY}, and 	$\tilde{\sigma}^2$ is obtained by taking $\vert \bm X_1\vert^2 = N$ into its definition below \eqref{binary AWGN channel}.	
Furthermore,	
${R}^{\max}_2 < 1$ is due to the Tag's binary modulation.
The other upper bound, $\log_2\left(1 + {N \vert g \vert^2\rho \myP}/{\sigma^2}\right)$, is due to the property that the capacity-achieving signal of the AWGN channel follows Gaussian distribution. This upper bound can only be mathematically achieved when the Tag is able to generate the Gausssian signal with zero mean and variance~$1$.

\subsection{Proof of Proposition~\ref{prop:high snr}}
Letting $\snr \rightarrow \infty$, we see that $h_1,\ h_0 \gg1$, $\underline{H},\ \overline{H} \rightarrow 1$, and hence $\overline{h}\gg 1$ and
\begin{equation}
\small
\overline{h} - \underline{h} =\! \frac{1}{2} \left\lvert h_1 -h_0 \right\rvert
\!=\!N \left\lvert \log_2\left(\frac{1 \!+\! \vert 1 +g \sqrt{\rho} c_1 \vert^2\snr}{1 \!+\! \vert 1 +g \sqrt{\rho} c_0 \vert^2\snr}\right)
\right\rvert
\!<\! \infty. 
\end{equation}
Thus, it is straightforward that $r_1 \rightarrow \infty$ and $r_2 \rightarrow 1 + \left(\overline{h} - \underline{h} \right)< \infty$ when $\snr \rightarrow \infty$, and~hence~$r_1>~r_2$, yielding the proof. 

\subsection{Proof of Proposition~\ref{prop:typical}} \label{appen:prop:typical}	
We have the following equation when $\vert g \vert^2$ is very small:
\begin{equation} \label{approx1}
\small
\begin{aligned}
h_i
&=N \log_2\left(1+ \left\lvert 1 +  g \sqrt{\rho} c_i\right\rvert^2 \snr\right)\\
&= N \left( \log_2\left(1 + \snr\right) 
+  L'\left(1+\SNR\right) \left(\vert g \vert^2 \rho \vert c_i \vert^2  + 2 \textrm{Real}\left[g \sqrt{\rho}c_i\right]\right) \snr 
+ o\left( \vert g \vert^2\right)\right),
\end{aligned}
\end{equation}
where $L'( x )$ is the first-order derivative function of the function $\log_2 (\cdot)$ at the point $x$. 
Based on \eqref{time sharing upperbound}, $\overline{H}$ is approximated as 
\begin{equation} \label{approx2}
\small
\overline{H} \approx \frac{1}{\ln 2} N \vert g \vert^2\rho \snr
\end{equation}
Based on \eqref{tse_equation} and \eqref{defi:mu}, we further have
\begin{equation} \label{approx3}
\small
\underline{H}  
\approx \frac{1}{2 \ln 2} \left(2 M(N)\mu\right)^2
\approx \frac{\vert g \vert^2 \rho}{2 \ln 2} \left(M(N) \vert c_1-c_{0}\vert \right)^2 \snr,
\end{equation}
where the first approximation in \eqref{approx3} is due to the fact that the term $\myexpect{\bm X_1}{\Perr} \approx 1/2$ in \eqref{tse_equation} when $\vert g \vert^2  \ll 1$, and the expansions as follows: 
\begin{equation}
\small
1- \Hb(p) = \frac{1}{2\ln 2} \sum\limits_{n=1}^{\infty} \frac{(1-2p)^{2n}}{n(2n-1)},\ 
M(N) \triangleq \frac{1}{2^N} \sum_{i=0}^{N-1}{N+i-1 \choose i} \frac{(N-i)}{2^i}.
\end{equation}
The second approximation in \eqref{approx3} is based on the approximation of \eqref{defi:mu} when $\vert g \vert^2  \ll 1$.

Therefore, taking \eqref{approx1}, \eqref{approx2} and \eqref{approx3} into \eqref{r1}, we have the following approximations when $\vert g \vert^2$ is very small:
\begin{equation}
\small
r_1  \approx \frac{\log_2\left(1 + \snr\right)}{\frac{1}{\ln 2} \vert g \vert^2 \rho \snr},\ 
r_2  \approx 
\frac{ L'(1 + \snr)  \left\lvert
	\vert g \vert^2 {\rho} \left(\vert c_1\vert^2 - \vert c_0\vert^2\right)
	+
	\vert g \vert \sqrt{\rho} \textrm{Real}\left[\left(c_1-c_{0}\right)e^{j\theta}\right]\right\rvert  }{\frac{1}{N} \frac{\vert g \vert^2 {\rho}}{2 \ln 2} \left(M(N) \vert c_1-c_{0}\vert \right)^2}.
\end{equation}
Therefore, $r_1, r_2 \gg 1$ and $r_1 \gg r_2$, when $\vert g \vert^2 \ll 1$.

\subsection{Proof of Proposition \ref{prop:BPSK-low-snr}}	\label{appen:prop:BPSK-low-snr}
When $\snr \rightarrow 0$, we have the approximation 
\begin{equation}
\label{h1-approx}
\small
h_i \approx \frac{N}{\ln 2} \left\vert 1+ g\sqrt{\rho} c_i\right \vert^2 \snr,\ i=0,1,
\end{equation}
and the approximations of $\overline{H}$ and $\underline{H}$ are given in \eqref{approx2} and \eqref{approx3}, respectively.

After taking \eqref{h1-approx}, \eqref{approx2} and \eqref{approx3} into \eqref{r1} and replacing $\left\lbrace c_1,c_0 \right\rbrace$ with $\left\lbrace 1,-1 \right\rbrace$ , we further have
\begin{equation} \label{r1r2}
\small
\begin{aligned}
&r_1-r_2 \approx  \max \left\lbrace \left\vert 1+ g\sqrt{\rho} \right \vert^2 , \left\vert 1- g\sqrt{\rho} \right \vert^2\right\rbrace
- \frac{N \ln 2 }{(2 M(N))^2} \left\vert \left\vert 1 + g\sqrt{\rho} \right \vert^2 -\left\vert 1 - g\sqrt{\rho} \right \vert^2 \right\vert - \vert g \vert^2 \rho.
\end{aligned}
\end{equation}	
Since the function of $N$, $\frac{N}{\left(M(N)\right)^2}$, decreases with $N$, $r_1-r_2$ increases with $N$.
In the following, we only need to prove $r_1-r_2>0$ when $N=1$.
Taking $N=1$ into \eqref{r1r2}, we have
\begin{equation} \label{mypi}
\small
r_1-r_2 
 \approx \left(1-\ln 2\right) \left\vert 1+ g\sqrt{\rho} \right \vert^2
+ \ln 2 \left\vert 1 - g\sqrt{\rho}\right \vert^2 - \vert g \vert^2 \rho.	
\end{equation}
By rewriting $g \sqrt{\rho}$ as $a + jb$ in \eqref{mypi}, we have
$r_1 -r_2 \approx a^2 + 2(1-2 \ln 2)a +1 + b^2 > 0$,
where the inequality can be easily verified from the fact that $a^2 + 2(1-2 \ln 2)a +1 > 0.85$.

\subsection{Proof of Proposition~\ref{prop:bpsk}}
	In this scenario, since $\vert 1 +g \sqrt{\rho} c_1 \vert =\vert 1 +g \sqrt{\rho} c_0 \vert$, we have $\overline{h} = \underline{h} =h_1=h_0$, and $r_2 =1$. Thus, the proof of the strict convexity is equivalent to the proof of $r_1 > 1$, i.e., $\overline{H}< h_1$, where 
$
h_1= N \log_2\left(1 + \frac{\left(1 + \vert g\vert^2 \rho \right)\myP}{\sigma^2}\right).
$
From \eqref{time sharing upperbound} and the property of the function $\log(\cdot)$, we have 
$
\overline{H} \leq \log_2\left(1 + \frac{N \vert g \vert^2 \rho \myP}{\sigma^2}\right) < h_1,\ \forall N\geq 1,
$
which completes the proof.

\subsection{Proof of Corollary~\ref{prop:nice}} \label{appen:prop:nice}
The asymptotic expression of the modified Bessel function $I_0(x)$ when $x \gg 1$ can be written as
$
I_0(x) = \frac{e^x}{\sqrt{2 \pi x}}
$~\cite{Bessel}.

Letting $f_{\chi^2}\left(x,2\right)=f_{\text{nc-}\chi^2}\left(x,2,\lambda\right)$ when $\lambda \gg 1$, we have 
$
\small
\frac{1}{2}e^{-\frac{x}{2}}
= \frac{1}{2} e^{-(x+\lambda)/2}
\frac{e^{\sqrt{\lambda x}}}{\sqrt{2 \pi {\sqrt{\lambda x}}}}.
$
After simplification, the equation above can be written as
$
\exp\left(-2\sqrt{\lambda x} \right) \left( -2\sqrt{\lambda x}  \right) = -\frac{1}{\pi e^{\lambda}}.
$
Based on the definition of Lambert W function~\cite{LambertW}, $W(\cdot)$, we further have
\begin{equation} \label{equation}
\small
-2\sqrt{\lambda x} = W\left(-\frac{1}{\pi e^{\lambda}}\right)
\end{equation}
The Lambert W functions has two branches, i.e., $W_0(x)$ and $W_{-1}(x)$, and $W_0(x) \rightarrow 0^-$ and $W_{-1}(x) \approx \ln(-x) \rightarrow -\infty$ when $x \rightarrow 0^-$~\cite{LambertW}, respectively.
For the upper branch, it is clear that the root of \eqref{equation} is zero, which, however, is not the root that increases with $\lambda$ as we want.
For the lower branch, we further simplify \eqref{equation} as 
$
-2\sqrt{\lambda x} = \ln\left(\frac{1}{\pi e^{\lambda}}\right),
$
thus, we have
$
2 \sqrt{\lambda x} = \lambda + \ln \pi.
$
Therefore, the root of \eqref{equation} is $\lambda/4$ when $\lambda \rightarrow \infty$.

\subsection{Proof of Corollary~\ref{prop:asynBER}} \label{appen:prop:asynBER}
When $\snr \rightarrow \infty$, from \eqref{prop:ber lb} and Propositions~\ref{prop:ber_bounds} and~\ref{prop:nice}, we can obtain 
$\PUBber = \PLBber + \frac{N}{2} \PUB$.
In the following, we only need to prove that $\PUB \ll \PLBber$ when $\snr \rightarrow \infty$.
Based on \eqref{prop:ber lb} and Proposition~\ref{prop:nice}, in the high SNR scenario, we have
$
\small
\PLBber = \frac{1}{2} \left(\exp\left(-\frac{\lambda}{8}\right) +
\left(1-Q_1\left(\sqrt{\lambda}, \frac{\sqrt{\lambda}}{2}\right)\right)
\right).
$
Since $Q(x) \triangleq \frac{1}{2} \text{erfc}\left(\frac{x}{\sqrt{2}}\right)$, adopting the useful   asymptotic expansion of $\text{erfc}(\cdot)$, i.e.,
\begin{equation}
\small
\text{erfc}\left(x\right) = \frac{e^{-x^2}}{x \sqrt{\pi}}\sum_{k=0}^{K-1} (-1)^k\frac{(2k-1)!!}{(2x^2)^k}
+ O\left(x^{1-2K}e^{-x^2}\right),
\end{equation} 
in the high SNR scenario, we have
$
\small
\PUB\!=\!Q\left(\!\sqrt{2 \snr }\left(\!\frac{1}{\sqrt{2}}-\vert g \vert \sqrt{\rho} \!\right)\!\right)
\!=\! \frac{1}{2} \frac{\exp\left(- \left(\sqrt{\snr }\left(\frac{1}{\sqrt{2}}-\vert g \vert \sqrt{\rho} \right)\right)^2\right)}{\sqrt{\pi} \left(\sqrt{\snr }\left(\frac{1}{\sqrt{2}}-\vert g \vert \sqrt{\rho} \right)\right)}.
$
Then, based on the assumption that $\vert g \vert \sqrt{\rho}\ll1$, it is easy to prove that 
\begin{equation}
\small
\PUB= Q\left(\sqrt{2 \snr }\left(\frac{1}{\sqrt{2}}-\vert g \vert \sqrt{\rho} \right)\right) \ll \frac{1}{2} \exp\left(-\frac{\lambda}{8}\right) < \PLBber.
\end{equation}

\ifCLASSOPTIONcaptionsoff
\fi
%

\end{document}